\documentclass[a4paper, 11pt]{article}
\usepackage[margin=1.7cm]{geometry}
\usepackage[utf8]{inputenc}
\usepackage[english]{babel}
\usepackage{array}
\usepackage{amsthm}
\usepackage{amsmath,amssymb,amsfonts,mathtools}
\usepackage{enumitem}
\usepackage{framed}
\usepackage{fullpage}
\usepackage{multirow}
\usepackage{accents}
\usepackage{xspace}
\usepackage{multirow,array}
\usepackage{xcolor}
\usepackage{calc}       % Required for \widthof
\usepackage{tikz-cd}
\usepackage{pdflscape}
\usepackage{graphicx}
\usepackage{float}

\usepackage{colortbl}
\usepackage{microtype}
\usepackage{hyperref}
\usepackage{booktabs,caption}
\usepackage[flushleft]{threeparttable}
\usepackage[T1]{fontenc}
\usepackage{tikz}
\usepackage{tikz-cd}
\usetikzlibrary{arrows}
\usepackage{braket}
\usepackage{multicol}
\usepackage{cleveref}
\usepackage{thm-restate}

\usepackage{tasks}
\usepackage{pdflscape}
\usepackage{calc}
\usepackage{quiver}
\usepackage{relsize}
\newsavebox{\complexitydiagram}
\usepackage[textsize=tiny]{todonotes}

\tikzset{
  treenode/.style = {align=center, inner sep=0pt, text centered,
    font=\sffamily},
  arn_n/.style = {treenode, circle, white, font=\sffamily\bfseries, draw=black,
    fill=black, text width=1.5em},% arbre rouge noir, noeud noir
  arn_r/.style = {treenode, circle, black ,font=\sffamily\bfseries, draw=black,  fill=white, text width=1.5em, very thick},% arbre rouge noir, noeud rouge
  arn_x/.style = {treenode, rectangle ,font=\sffamily\bfseries, draw=red,
    minimum width=1.5em, minimum height=1.5em}% arbre rouge noir, nil
}
\hypersetup{
    pdftitle={}, % title
    pdfauthor={}, % author
    colorlinks=true, % false: boxed links; true: colored links
    linkcolor=blue, % color of internal links
    citecolor=blue, % color of links to bibliography
    filecolor=blue, % color of file links
    urlcolor=blue % color of external links
}

\hypersetup{colorlinks,linkcolor={blue},citecolor={blue},urlcolor={red}}

\hypersetup{colorlinks,linkcolor={blue},citecolor={blue},urlcolor={red}}

\newlength{\dhatheight}

\definecolor{lgreen}{rgb}{0.6, 1.0, 0.0}

\DeclarePairedDelimiter\card{\lvert}{\rvert}
\DeclarePairedDelimiter\bracket{[}{]}
\DeclarePairedDelimiter\of{(}{)}
\DeclarePairedDelimiter\norm{\lVert}{\rVert}

\let\set\undefined
\DeclarePairedDelimiter\set{\{}{\}}
\newcommand{\comprehension}[2]{\set*{#1\,\middle|\,#2}}

\newcommand{\for}[1]{{\forall\:\!#1,\,}}

\renewcommand{\geq}{\geqslant}
\renewcommand{\leq}{\leqslant}
\renewcommand{\Pr}{\mathbb{P}}

\newcommand{\Ra}{\mathrel\Rightarrow}

\newcommand{\predicate}[1]{{\textup{#1}}} % non-italic text for predicates
\renewcommand{\text}[1]{{\textit{#1}}} % italic text (independantly of the environment) for other use cases: word in index or superscript, name of variable...

\newcommand{\limitA}[2]{#1\limits_{#2}}
\newcommand{\limitB}[3]{\mathop{\limitA{#1}{#2}}\limits_{#3}}

\newcommand{\A}{\mathcal{A}}

\newcommand{\C}{\mathcal{C}}
\newcommand{\D}{\mathcal{D}}

\renewcommand{\H}{\mathcal{H}}

\newcommand{\NN}{\mathbb{N}}

\newcommand{\zone}{{\{0, 1\}}}

\DeclareMathAlphabet{\mathbbfrombbold}{U}{bbold}{m}{n}

\newcommand{\One   }{\mathbbfrombbold{1}}%index function

\newcommand{\cert}{\mathsf{cert}}%to be put in exponent
\newcommand{\ind}{\mathsf{ind}}%to be put in exponent
\newcommand{\str}{\mathsf{str}}%to be put in index
\newcommand{\weak}{\mathsf{weak}}%to be put in index
\newcommand{\sab}{\mathsf{sab}}%to be put in index

\renewcommand{\deg}{\mathrm{deg}}

\newcommand{\ndeg}{\mathrm{ndeg}}
\newcommand{\rdeg}{\mathrm{rdeg}}
\newcommand{\adeg}[1][]{\widetilde{\mathrm{deg}#1}}
\newcommand{\osdeg}[1][]{\widetilde{\mathrm{osdeg}#1}}
\newcommand{\bmdeg}{\mathrm{bm}\predicate{-}\mathrm{deg}}
\newcommand{\bmadeg}[1][]{\widetilde{\mathrm{bm}\predicate{-}\mathrm{deg}#1}}

\newcommand{\UC}{\mathsf{UC}}
\newcommand{\CG}{\mathsf{CG}}

\newcommand{\dqc}{\mathsf{D}}
\newcommand{\qec}{\mathsf{Q}_{\mathsf{E}}}
\newcommand{\roc}{\mathsf{R}_{\mathsf{0}}}
\newcommand{\qoc}{\mathsf{Q}_{\mathsf{0}}}

\newcommand{\rqc}{\mathsf{R}}
\newcommand{\qqc}{\mathsf{Q}}
\newcommand{\cc}{\mathsf{C}}
\newcommand{\s}{\mathsf{s}}
\newcommand{\bs}{\mathsf{bs}}

\newcommand{\EC}{\mathsf{EC}}
\newcommand{\RC}{\mathsf{RC}}
\newcommand{\QC}{\mathsf{QC}}
\newcommand{\RS}{\mathsf{RS}}
\newcommand{\QS}{\mathsf{QS}}

\newcommand{\QD}{\mathsf{QD}}
\newcommand{\MM}{\mathsf{MM}}
\newcommand{\fbs}{\mathsf{fbs}}
\newcommand{\cfbs}{\mathsf{cfbs}}
\newcommand{\FC}{\mathsf{FC}}

\newcommand{\CAdv}{\mathsf{CAdv}}
\newcommand{\Prt}{\mathsf{Prt}}  %exponentially scaled mesure
\newcommand{\prt}{\mathsf{prt}} % prt=log(Prt)
\newcommand{\NP}{\mathsf{NP}}
\newcommand{\OR}{\mathsf{OR}}

\newcommand{\Tribes}{\mathsf{Tribes}}
\newcommand{\AND}{\mathsf{AND}}

\newcommand{\Exact}{\mathsf{Exact}}

\renewcommand{\hat}{\widehat}

\declaretheorem{theorem}
\declaretheorem[sibling=theorem]{corollary}
\declaretheorem[sibling=theorem]{remark}

\declaretheorem[sibling=theorem]{lemma}

\declaretheorem[name=Definition,sibling=theorem]{defi}

\declaretheorem[sibling=theorem]{note}

\declaretheorem[sibling=theorem]{fact}
\declaretheorem[sibling=theorem]{proposition}

\usepackage[
backend=bibtex,
style=alphabetic,
sorting=nyt,%ynt
maxbibnames=99,
minalphanames=3
]{biblatex}
\title{Certification complexity of Boolean functions}

\author{Chandrima Kayal\thanks{Universit\'{e} Paris-Cit\'{e}, CNRS, IRIF  \texttt{kayal@irif.fr}} 
\and 
Sophie Laplante\thanks{ Universit\'{e} Paris-Cit\'{e}, IRIF \texttt{laplante@irif.fr}}
\and
\'{E}mile Larroque\thanks{ Universit\'{e} Paris-Cit\'{e}, IRIF \texttt{larroque@irif.fr}}
\and
Kri{\v s}j{\=a}nis Pr{\=u}sis\thanks{University of Latvia, Faculty of Science and Technology \texttt{krisjanis.prusis@lu.lv}}
\and
Jevg{\=e}nijs Vihrovs\thanks{University of Latvia, Faculty of Science and Technology \texttt{jevgenijs.vihrovs@lu.lv}}
}

\date{}

\begin{document}

\date{}

\maketitle
\thispagestyle{empty}
\begin{abstract}
Certificate complexity ($\cc(f)$) is a fundamental measure of complexity of Boolean functions $f$ which counts the number of bits of an input that need to be known in order for the value of the function to be determined. A certificate can be viewed as a partial assignment, or a boolean subcube where the function is constant.
Certificate complexity is well understood for deterministic query (or decision tree) complexity ($\dqc$) and other  query models such as bounded-error randomized and quantum complexity ($\rqc$, $\qqc$), but not as well for quantum zero-error ($\qoc$) and exact query complexity ($\qec$), where there is no agreed-upon certificate ``object'' (even for $\qqc$).

Instead, we study an operational notion of certification and apply it to various query-based models, with a focus on zero-error and exact quantum query complexity, but also on polynomial degree measures.
We give new characterizations of $\cc$, $\RC$ (randomized certificate complexity) and $\QC$ (quantum certificate complexity), in terms of various measures such as classical and quantum sabotage complexity, unambiguous certificate complexity, and variants of polynomial degree.

Certification complexity also gives rise to new lower bounds on $\qec$ and $\qoc$, the quantum analogues of $\dqc$ and $\roc$, complexity measures for which few lower bound techniques are known which are not already lower bounds for two-sided error quantum query complexity.
{We exhibit a total Boolean function for which our certification complexity measure gives a tight lower bound for $\qoc$, but rational degree and $\qqc$ are asymptotically smaller.}

\end{abstract}

\bigskip
\begin{multicols}{2}
    \tableofcontents
\end{multicols}
\bigskip
\bigskip

\section{Introduction}

A certificate captures the amount of information about an input that is sufficient to ascertain the value of a function on a fixed input $x$. This notion appears naturally in several computational models. In the query model, a certificate for an input is a partial assignment of the input's variables such that every completion of this partial assignment yields the same function value. 
Thus,
a certificate can be viewed as a set of queried variables whose values are sufficient to certify an output. A similar notion arises in communication complexity. Here, the input is distributed among two parties, and a certificate corresponds to a monochromatic rectangle: a product set of inputs on which the function has a fixed value. In both settings, the underlying idea is the same: rather than revealing the entire input, we seek a restricted region of the input space that is already sufficient to determine the output. The complexity of a certificate measures the amount of data required to identify such a monochromatic region.

In this paper, we systematically study an abstract notion of certification that applies to various complexity measures in the query complexity world. The main object of study are \emph{decision problems}, that is, (possibly partial)  functions $f : \D \subseteq \zone^n \to \zone$ and we are interested in the complexity of the certification of $f$ with respect to different complexity measures. In the classical deterministic setting, the minimum number of queries needed for any deterministic query algorithm to output $f(x)$ for the worst case input $x$ is known as \emph{deterministic query complexity} of $f$, denoted $\dqc(f)$. The notion of certificate is especially clear for $\dqc(f)$: A certificate is a \emph{partial assignment}
$\rho : [n] \to \set*{0,1,\ast}$,
that forces the value of the function to be constant. $\cc(x, f)$ is the size of the smallest certificate that is consistent with $x$, and 
$\cc(f) = \max_x\{\cc (f,x)\}$.
Equivalently, it corresponds to a monochromatic subcube containing the input. 

In randomized and quantum computation, the notion becomes less immediate: what should constitute a certificate when the computation itself may involve randomness or querying in superposition? 
The quantum case is especially troublesome since a in a quantum query algorithm, we cannot assume that the queries that were made are known at the end of the computation. 
Hence, it is not clear what the corresponding certificate should be as a concrete combinatorial object. Although quantum certificate complexity is well understood for two-sided error, as a complexity measure, what would be the natural quantum equivalent to a classical certificate object is still a matter of debate~\cite{Jeffery2026QuantumQuery,KKW+26}. Several notions capture different aspects of classical certification, but no single quantum measure appears to play the canonical role of certificate complexity. An interesting perspective comes from postselection: the exact postselected randomized query complexity coincides with classical certificate complexity~\cite[Theorem~16]{Cade20}, whereas its quantum counterpart is characterized by rational degree. This suggests that rational degree can be viewed as a natural quantum analogue of certificate complexity. However, this notion is distinct from Aaronson's quantum certificate complexity~\cite{Aar08}, which captures a different operational notion of quantum certification. The presence of these different, non-equivalent candidates highlights a more fundamental question: what is the appropriate notion of a certificate in the quantum query model, and which aspects of classical certification should such a notion preserve?

Rather than insisting on a model-specific notion of a certificate ``object'', we take a more operational viewpoint and focus on the certification \emph{task}.  
Similar notions have appeared in various forms in previous works~\cite{Aar08, KT16,BK19}.
Given a fixed, publicly known input $z$ (so $f(z)$ is also known, as $f$ is fixed), the certification task for $z$ is as follows. On an arbitrary input $x$,
output ``accept'' whenever
$x=z$,
and output ``reject'' whenever
$f(x) \neq f(z)$.

What it means to compute this task will depend on the model of computation we are certifying for. 
Notice that since we are computing~$f$ on a restricted domain, this means that (up to negation) certification  for a model $\mathsf{M}$ is an easier task than computing $f$ in full, so certification complexity for $\mathsf{M}$ is usually a lower bound on  $\mathsf{M}$.
In the case of deterministic query complexity, the certification task and the certificate objects are equivalent: making the queries to indices in the certificate suffices to achieve the certification task, and conversely, the queries made by an accepting execution of a certification algorithm form a certificate for this input.
In this setting, certification is also equivalent to nondeterministic computation, or postselection, where it suffices to guess a certificate, and make those queries, to check for consistency with the given input.

Taking an operational perspective leads  us to study certification across computational models without requiring {\em a priori} a common notion of what a certificate object should look like.  
We systematically investigate the certification task for various complexity measures in the query world, including quantum query complexity, polynomial degree, and several related measures.

\subsection{Certification  in operational terms}

Let us define our object of study formally.

\begin{restatable}[input certification problem, or $z$-certification version of a function]{defi}{fzCertDefiniton}\label{def:f cert}
    Let
    $f : \D \subseteq \zone^n \to \zone$
    be a Boolean (possibly partial) function and
    $z \in \D$.
    Then $f^\cert_z$ denotes the \emph{$z$-certification version of~$f$}, which is a partial function defined on the domain
    %$x \in 
    $\set*{z} \cup \comprehension{y \in \D}{f(y) \neq f(z)}$
    as
    $$f^\cert_z(x) := \begin{cases}
        1
        &\predicate{if } z=x\\
        0
        &\predicate{if } f(z) \neq f(x)
    \end{cases}$$
\end{restatable}

\begin{restatable}[certification complexity $\mathsf{M}^\cert$]{defi}{certDefinition}\label{def:cert qc certification version}
    Consider a complexity measure $\mathsf{M}$ of Boolean (possibly partial) functions. Let
    $f : \D \subseteq \zone^n \to \zone$
    be a Boolean (possibly partial) function. $\mathsf{M}^\cert(f)$ denotes the \emph{certification complexity of $f$} obtained from $\mathsf{M}$. $\mathsf{M}^\cert(f)$ is defined as the complexity of the (worst input $z$ case) $z$-certification version of $f$, that is, 
    $$\mathsf{M}^\cert(f) := \max\limits_{z \in \D} \mathsf{M}\of*{f^\cert_z}.$$
\end{restatable}

As an example, consider the decision tree model $\dqc$. Certification for $z$ (which is fixed and known) can be achieved by making queries to an input $x$ according to a decision tree for $f$, and outputting $1$ if for all queries $i$ that were made, $x_i=z_i$, and $0$ otherwise. Clearly, this procedure will output $1$ when $x=z$, and will output $0$ whenever $f(x)\neq f(z)$. Running this procedure on $z$ therefore certifies the value of $f(z)$. Since the domain of $f^\cert_z$ is a subset of $\D$, certifying is no harder than computing the function on the full domain, so we have that $\dqc^\cert(f) \leq \dqc(f)$.

To illustrate the difference between decision and certification, consider the $\Tribes$ function. Its input is composed of $m$ blocks $X_1,\ldots,X_m$ where each $X_i$ is an $l$-bit string, and  $\Tribes_{l,m}\of*{X_1,\ldots, X_m} := \bigwedge_{i=1}^m\OR\of*{X_i}$. Certifying $z$ reduces the domain to $z$ and all inputs $x$ where $f(x)\neq f(z)$. Whenever $\Tribes_{l,m}(z)=0$, $z$ contains an all-$0$ block, so to certify $z$, it suffices to output $1$ on all inputs where this block is all-$0$, and zero everywhere else, essentially solving $\OR_l$ on this block.
Similarly, if $\Tribes_{l,m}(z)=1$, then~$z$ has at least some variable in each block that is set to $1$ in $z$. To certify $z$, it suffices to output $1$ whenever each of these variables is $1$, and $0$ on all other inputs. This amounts to deciding $\AND_m$ on these variables. Depending on the measure, this could be a substantially simpler task.

Operational definitions of certification have appeared in the literature in various forms~\cite{Aar08,BK19}. For example, Ben-David and Kothari~\cite[Definition~7]{BK19} state that one can understand the certificate complexity $\cc$ and its variants $\QC$ and $\RC$ through a unifying computational characterization.

\subsection{Our contributions}

Our main results fall into three categories:  new characterizations of known complexity measures and 
the collapse of various certification measures,  new lower bounds and a new framework to study quantum certificates, and some separations.

We  apply the certification framework to a wide variety of complexity measures. We defer the formal definitions of these measures to \Cref{sec:preliminaries}.

%%#################################################
\subsubsection{Characterization of certificate complexity: classical and quantum}
We start by showing that the input certification approach leads to several new characterizations of certificate complexity $\cc$ and its variants $\QC$ and $\RC$, an indication of the robustness of these notions.

\begin{restatable}[characterizations of classical certificate complexities]{theorem}{CharacterizationClassicalCertificateComplexities}\label{characterizations of C and RC}
    For all Boolean (possibly partial) function $f$
    \begin{enumerate}
        \item\label{item:characterizations of C} $\cc(f) = \dqc^\cert(f) = \roc^\cert(f) = \UC^\cert(f) = \Theta\of*{\CG^\cert(f)}$
        \item\label{item:characterizations of RC} up to $O(1)$ multiplicative factors,
        $$\RC(f) = \rqc^\cert(f) = \RS^\cert(f) = \FC(f) = \EC^\cert(f) = \prt^\cert(f) = \fbs(f) = \cfbs^\cert(f)$$
    \end{enumerate}
where $\UC$ is the unambiguous certificate complexity (\Cref{def:UC}), $\RS$ is the sabotage complexity (\Cref{def:sabotage complexity}),
$\CG$ is the private coin certificate game complexity (\Cref{def:CGpriv}),
$\EC$ is the expectational certificate complexity (\Cref{def:EC}), $\prt$ is the partition bound (\Cref{def:prt}) and $\cfbs$ is the critical fractional block sensitivity (\Cref{def:cfbs}).
\end{restatable}

For completeness, we included in the statement of \Cref{characterizations of C and RC} the previously known characterizations
$\FC(f) = \fbs(f) = \Theta\of*{\RC(f)}$
(see \Cref{local cert = local} for references).

Then we move to the quantum case, starting with characterizations of the bounded-error quantum query complexity $\qqc$.

\begin{restatable}[computational characterizations of quantum certificate complexity]{theorem}{ComputationalCharacterizationQuantumCertificateComplexity}\label{characterizations of QC: computational}
    Up to $O(1)$ multiplicative factors, for all Boolean (possibly partial) functions $f$,
    $$\QC(f) = \qqc^\cert(f) =  \QD^\cert(f) = \MM^\cert(f) = \QS^\cert(f)$$
    where
    $\QD$ is the quantum distinguishing complexity~\cite{BK19} (\Cref{def:QD}),
    $\MM$ is the quantum adversary bound (minimax formulation, \Cref{def:MM}),
    $\QS$ is the quantum sabotage complexity~\cite{CMP24} defined in \Cref{sec:quantum cert measures}.\footnote{The theorem holds for the four variants of quantum sabotage complexity, and details are provided in \Cref{sec:quantum cert measures}.}
\end{restatable}

\subsubsection{Exact and zero-error quantum certification complexity}
A zero-error quantum algorithm (or a Las Vegas quantum algorithm) never produces an incorrect
answer on an input, but is allowed to claim ignorance with probability
at most $\frac{1}{2}$. The zero-error quantum query complexity of $f$, denoted $\qoc(f)$, is the minimum number of
queries needed for a zero-error quantum algorithm to compute $f$. On the other hand, $\qec$ denotes the minimum
number of queries in a quantum algorithm that computes $f$ exactly. It follows from the definitions that $\qqc(f) \leq \qoc(f) \leq \qec(f)$ for any function $f$. For a formal definition, see \Cref{def:qqc}, \Cref{def:Q0} and \Cref{def:qec}. Our first contribution is a lower bound on both the complexity measures $\qoc$ and $\qec$. It turns out the certification versions coincide while the certification for bounded error quantum query complexity $\qqc$ remains different.

\begin{restatable}[exact/zero-error collapse of $\qoc$ and $\qec$ through the $\cdot^\cert$ transform]{theorem}{QZeroCertEqualsQECert}\label{Q_0 cert = Q_E cert}
    Up to $O(1)$ multiplicative factors, for any Boolean (possibly partial) function $f$
    $$\qec^\cert(f) = \qoc^\cert(f) \leq \qoc(f).$$
\end{restatable}

This result can be seen as the quantum version of the equality $\cc(f) = \dqc^\cert(f) = \roc^\cert(f)$, since $\qec$ is the quantum analogue of $\dqc$, and $\qoc$ is the analogue of $\roc$, although the proof of \Cref{Q_0 cert = Q_E cert} requires a more involved argument that uses exact quantum amplitude amplification (see \Cref{sec:equality Q_0 cert = Q_E cert} for details).

To the best of our knowledge, 
$\qec^\cert$
is a new lower bound on quantum zero-error query complexity $\qoc$, a complexity measure that, aside from rational degree, lacks lower bound techniques that are not also lower bounds on $\qqc$.
The fact that it has two equivalent formulations can also be viewed as an indication that this is a robust notion that merits  further study.

\subsubsection{Certification using polynomial degree}\label{subsec:polynomial-certification}

Perhaps the most intriguing aspect of our systematic study of certification is its application to non-querying  complexity measures, including polynomial degree, sensitivity, and related measures.
Applying certification to polynomials might seem at odds with the query based nature of certification, but it results in some interesting and often surprising connections with query based certification measures.

\paragraph{Polynomial lower bounds for quantum query complexity}
We start by giving two new characterizations of quantum certificates in terms of polynomials.

\begin{restatable}[polynomial characterizations of quantum certificate complexity]{theorem}{PolynomialCharacterizationQuantumCertificateComplexity}\label{characterizations of QC: polynomial}
    Up to $O(1)$ multiplicative factors, for all Boolean (possibly partial) functions $f$,
$$\qqc^\cert(f) = \adeg^\cert(f) = \bmadeg^\cert(f).$$
\end{restatable}

Here $\bmadeg$ denotes block-multilinear approximate bounded-polynomial degree (\Cref{def:bmdeg}).

This is a good example of what we call the ``Sandwich Lemma'' (\cref{A > B > A cert = B cert} from \Cref{toolbox})
which says that if $\mathsf{M}^\cert \leq \mathsf{N} \leq \mathsf{M}$, then $\mathsf{N}^\cert$ collapses to $\mathsf{M}^\cert$.
The theorem follows because  $\qqc^\cert \leq \adeg \leq \bmadeg\leq \qqc$~\cite{BBCM+01,AA15,AAI+16,ABK21}.

\paragraph{Polynomial degree lower bound for zero-error quantum query complexity}
%\sophie{ please check *}
Next, we define one-sided-error approximate polynomials and the corresponding degree complexity measure $\osdeg$ (\Cref{def:deg measures}), providing a new lower bound on zero-error quantum query complexity that is stronger than both approximate degree (a lower bound on $\qqc$), and rational degree (a lower bound on both exact polynomial degree and $\qoc$). We show that its certification version coincides with exact degree certification complexity~$\deg^\cert$. This might seem surprising, but on the other hand, it can be seen as the polynomial degree equivalent of the collapses we have seen of $\dqc^\cert$ with $\roc^\cert$, and of $\qec^\cert$ with $\qoc^\cert$.

\begin{restatable}[bounded-polynomial degree measures]{defi}{degDefinition}\label{def:deg measures}
    Let
    $\varepsilon_0,\varepsilon_1 \geq 0$
    be error parameters and let
    $\beta \geq 1$
    be a bounding parameter.
    Let
    $f : \D \subseteq \zone^n \to \zone$
    be a Boolean (possibly partial) function.
    The \emph{$(\varepsilon_0,\varepsilon_1)$-approximate $\beta$-bounded-polynomial degree of $f$}, denoted $\Delta_{\varepsilon_0,\varepsilon_1,\beta}(f)$, is defined as
    \begin{align*}
        \Delta_{\varepsilon_0,\varepsilon_1,\beta}(f) \qquad
        &\predicate{is the }\min\predicate{ of}\quad \deg(P)\\
        &\predicate{over multilinear real polynomials }P\\
        &\predicate{subject to error conditions}\quad \for{x \in f^{-1}(0)} \card*{P(x) - f(x)} \leq \varepsilon_0\\
        &\predicate{and}\quad \for{x \in f^{-1}(1)} \card*{P(x) - f(x)} \leq \varepsilon_1\\
        &\predicate{and to the bounding condition}\quad \for{x \in \zone^n} P(x) \in \bracket*{0,\beta}
    \end{align*}

    Define the (exact) \emph{bounded-polynomial degree of $f$}, denoted $\deg(f)$,
    the \emph{$\varepsilon$-approximate bounded-polynomial degree of $f$}, denoted $\adeg[_\varepsilon](f)$,
    and the \emph{approximate bounded-polynomial degree of $f$}, denoted $\adeg(f)$, as
    $$\deg(f) := \Delta_{0,0,1}(f) \qquad\qquad
    \adeg[_\varepsilon](f) := \Delta_{\varepsilon,\varepsilon,1}(f) \qquad\qquad
    \adeg(f) := \Delta_{\frac{1}{3},\frac{1}{3},1}(f).$$
    
    For any function value
    $b \in \zone$,
    and error parameter
    $\varepsilon \geq 0$,
    define the \emph{$b$-sided-error $\varepsilon$-approximate bounded-polynomial degree of $f$}, denoted $\osdeg[^b_\varepsilon](f)$,
    and the \emph{$b$-sided-error approximate bounded-polynomial degree of $f$}, denoted $\osdeg[^b](f)$,as
    $$\osdeg[^0_\varepsilon](f) := \Delta_{\varepsilon,0,1}(f) \qquad\qquad
    \osdeg[^1_\varepsilon](f) := \Delta_{0,\varepsilon,1}(f) \qquad\qquad
    \osdeg[^b](f) := \osdeg[^b_\frac{1}{3}](f)$$
    and define the \emph{one-sided-error $\varepsilon$-approximate bounded-polynomial degree of $f$}, denoted $\osdeg[_\varepsilon](f)$,
    and the \emph{one-sided-error approximate bounded-polynomial degree of $f$}, denoted $\osdeg(f)$, as
    $$\osdeg[_\varepsilon](f) := \max\of*{\osdeg[^0_\varepsilon](f),\osdeg[^1_\varepsilon](f)} \qquad\qquad\qquad
    \osdeg(f) := \osdeg[_\frac{1}{3}](f).$$
\end{restatable}

Note that the choice to set the bounding parameter $\beta$ to $1$ is without loss of generality, as setting it to a different \emph{constant} value leads to the same measures $\deg$, $\adeg$, $\osdeg[^0_\varepsilon]$, $\osdeg[^1_\varepsilon]$ and $\osdeg[_\varepsilon]$ up to a constant factor depending only on $\varepsilon$ and $\beta$ on the degree, and without any change in the error parameter. This fact corresponds to the lemmas proved in Appendix~\ref{appendix:boundedness}.

We introduce one-sided-error approximate bounded-polynomial degree as a new lower bound on zero-error quantum query complexity.

\begin{restatable}{theorem}{LowerBoundOnQZeroByOsdeg}\label{osdeg < Q_0}
    For any
    $0 < \varepsilon < 1$
    and any Boolean (possibly partial) function
    $f:\D \subseteq \zone^n \to \zone$,
    $$\qoc(f) \geq \Omega\of*{\osdeg[_\varepsilon](f)}.$$
\end{restatable}

Although several variants of one-sided polynomials have appeared in the literature, to the best of our knowledge, this definition does not coincide with any of them. 
Since it is larger than both rational degree and bounded-polynomial approximate degree, it could lead to better lower bounds on $\qoc$. 
Moreover, proving lower bounds on one-sided-error approximate degree is comparatively easier than working with the completely bounded polynomial formulation of Arunachalam, Briët, and Palazuelos~\cite{ABP19}, which characterizes quantum query complexity while also capturing the error parameter.

\begin{restatable}[certification equivalence between $\deg$ and $\osdeg$]{theorem}{CollapseDegCertOsdegCert}\label{deg cert = osdeg cert}
    For any
    $0 < \varepsilon < 1$,
    up to $O(1)$ multiplicative factors, for any Boolean (possibly partial) function $f$
    $$\deg^\cert(f) = \osdeg[_\varepsilon]^\cert(f) \leq \qoc^\cert(f).$$
\end{restatable}

\paragraph{Certification for degree is not rational degree}

{
Each one of \Cref{Q_0 cert = Q_E cert} and \Cref{deg cert = osdeg cert} independently implies that $\deg^\cert$ is a polynomial degree lower bound method on $\qoc$.
In \Cref{sec: ndeg degcert separation}, we exhibit a total Boolean function $f_m$ for which
$\qoc\of*{f_m} = {\widetilde\Theta}\of*{\deg^\cert\of*{f_m}} = {\widetilde\Theta}\of*{\qec^\cert\of*{f_m}}$,
and both
$\qqc\of*{f_m}, \rdeg\of*{f_m}$ are asymptotically smaller than $ \deg^\cert\of*{f_m}$.

\begin{restatable}{theorem}{SeparationRdegDegCert}
\label{thm:separation (deg cert)^2.5 < ndeg}
    There exists a total function $f_m$ on $m^2$ bits such that 
    \begin{enumerate}
    \item $\rdeg\of*{f_m} = m$  
    \item $\deg^\cert\of*{f_m} \geq {\widetilde\Omega}\of*{m^{3/2}}$
    \end{enumerate}
\end{restatable}

This shows that the degree certification measure and the quantum exact certification complexity are new lower bound methods on $\qoc$.
We show that for this function, the bound is tight for $\qoc$, and that $\qqc$ is asymptotically smaller that $\qoc$ 
(\Cref{thm:Q_0 algo}). 
}

\subsubsection{Polynomial relations between  measures}\label{subsec:polynomial relations}

Spectral sensitivity (\Cref{def:lambda}) is a  lower bound on sensitivity, and was shown by Huang~\cite{Hua19, ABK+21}  to be   polynomially related to known query complexity measures for total functions.

Our first result is that certification for spectral sensitivity equals the square root of sensitivity, which establishes that it is also polynomially related to the other known bounds. 

We also show that  rational degree (\Cref{def:rdeg}), a  lower bound on $\qoc$, in equal to its certification measure. Rational degree was recently shown to be polynomially related to query complexity and other known measures, so these two results can be interpreted to mean that certification produces nontrivial measures for the purposes of proving lower bounds.

\begin{restatable}[polynomial relationship for certification measures]{theorem}{SanityCheck}\label{sanity check}
    For any Boolean (possibly partial) function $f$
    \begin{tasks}[label={\arabic*.}](2)
        \task $\lambda^\cert(f) = \sqrt{\s(f)} \geq \sqrt{\lambda(f)}$ (\Cref{lambda cert = sqrt(s)})
        \task $\rdeg^\cert(f) = \rdeg(f)$ (\Cref{ndeg cert = ndeg})
    \end{tasks}
\end{restatable}

To conclude, we provide a quadratic relationship between block sensitivity and $\deg^\cert$, and a
cubic relationship between certificate complexity and $\deg^\cert$ for total functions.

\begin{restatable}{theorem}{UpperBoundOnCByDegCertBs}%
\label{bs < (deg cert)^2 and C < (deg cert)^3}
    Up to $O(1)$ multiplicative factors,
    \begin{enumerate}
        \item\label{item: bs < (deg cert)^2} for any Boolean (possibly partial) function $f$, $\bs(f)\leq \deg^\cert(f)^2$
        \item\label{item: C < (deg cert)^3} for any Boolean total function $\cc(f) \leq \deg^\cert(f)^3$.
    \end{enumerate}
\end{restatable}

We give a proof sketch to illustrate some of our proof techniques.
For \cref{item: bs < (deg cert)^2}, 
Nisan and Szegedy showed~\cite{NS94}
$$\bs(f) \leq 2 \deg(f)^2$$
and their result holds for partial functions. Therefore we can apply one of our Toolbox results described in the next section (\cref{cert monotonicity} from \Cref{toolbox}) and get
$$\bs^\cert(f)\leq2 \deg^\cert(f)^2.$$
Combining this with the fact that
$\bs^\cert(f)=\bs(f)$
(\Cref{local cert = local}) concludes the proof.

For \cref{item: C < (deg cert)^3}, Midrijanis  showed~\cite[Theorem~4]{midrijanis2004exact} that for total functions $f$,
$$\dqc(f) \leq  \deg(f) \cdot\bs(f).$$
Although there are partial functions for which this theorem fails to hold, it  holds for any certification function $f^\cert_z$, provided that $f$ is total: Lemma~3 in \cite{midrijanis2004exact} shows that for a total $f$, every maxonomial of a polynomial representing $f$ contains a sensitive block for $z$. This block remains sensitive on $z$ when we restrict the function to $f^\cert_z$, since $f^{-1}(1{-}f(z))$ is still included in the domain of $f^\cert_z$. Then, for Theorem~4 from that paper, if the queried bits of the input have not fixed the value of the function, they have to be consistent with $z$. Therefore the queried monomials have to contain a sensitive block for $z$, and the rest of the proof goes through. Thus, we get
$$\dqc\of*{f^\cert} \leq  \deg\of*{f^\cert_z} \cdot \bs\of*{f^\cert_z}$$
which is the same as
$$\dqc^\cert(f) \leq  \deg^\cert(f) \cdot\bs^\cert(f).$$

Combining this with \cref{item: bs < (deg cert)^2} and the fact that
$\bs^\cert(f)=\bs(f)$
(\Cref{local cert = local}) and
$\dqc^\cert(f) = \cc(f)$
(\Cref{characterizations of C and RC}),
we get the claimed bound on $\cc$.

Notice that \cref{item: bs < (deg cert)^2} is subsumed by the result of de~Wolf that
$\cc(f) %\leq \max\set*{\ndeg(f) \cdot \bs_0(f),\ndeg(\neg f) \cdot \bs_1(f)}
\leq \max\set*{\ndeg(f),\ndeg(\neg f)} \cdot \bs(f)$ \cite[Lemma~2.6]{dW03}
together with the relations
$\rdeg(f) = \max\set*{\ndeg(f),\ndeg(\neg f)}$~\cite[Fact~3]{KKW+26}
and
$\rdeg \leq \osdeg^\cert = \deg^\cert$
(\Cref{ndeg < osdeg cert} and \Cref{deg cert = osdeg cert}). See \Cref{sec:def deg measures} for the definitions of $\rdeg$ and $\ndeg$.

Nevertheless, the proof highlights an interesting property of the class of partial functions $f^\cert_z$ where $f$ is a total: it is an example of partial functions where the upper bound of Midrijanis still holds.
This class of partial functions might have other interesting properties. For example,  Huang's resolution  of the sensitivity conjecture holds for total functions, and there are partial functions for which it fails. 
Therefore, one might ask:  does it hold
for this family of partial functions?
If so, it would imply
$\FC(f) \leq {\deg^\cert(f)}^2 \leq {\lambda^\cert(f)}^4 = {\s(f)}^2$
for all Boolean total functions $f$, which would be tight~\cite{Rub95},
whereas the best upper bound known to hold for total functions is $\FC(f) \leq {\s(f)}^4$.

\subsection{Key properties of certification measures}

We make  use of the following general properties of certification measures $\mathsf{M}^\cert$ as a central toolbox to prove many of our results (\Cref{toolbox}, \Cref{fig:toolbox}).

We say that a measure $\mathsf M$ is {\em stable} if the following two properties hold for any (possibly partial) function $f$ :
\begin{enumerate}
\item\label{item stable:negation} $\mathsf M(f) = \mathsf M(\neg f)$
\item\label{item stable:monotonicity under restriction} for any $g$ that is a restriction of $f$ to a  domain $\D' \subseteq \D$, $\mathsf M(g) \leq \mathsf M(f)$.
\end{enumerate}

\newcounter{enumCounterToolbox}
\begin{restatable}[certification complexity toolbox]{theorem}{Toolbox}
\label{toolbox}
    Let $\mathsf{M}$ and $\mathsf{N}$ be complexity measures defined on all Boolean (possibly partial) function
    $f : \D \subseteq \zone^n \to \zone$. Then the following holds.
    \begin{enumerate}[label=(\roman*), ref=(\roman*)]
    
        \item\label{cert monotonicity} (the certification transform is monotone) If $\mathsf{M}(g)\leq\mathsf{N}(g)$ for every Boolean (possibly partial) function~$g$, then
        $$\mathsf{M}^{\cert}(f)\leq\mathsf{N}^{\cert}(f).$$
        
        \setcounter{enumCounterToolbox}{\value{enumi}}
    \end{enumerate}
    Furthermore, if $\mathsf M$ is stable, then the following properties also hold.
    \begin{enumerate}[label=(\roman*), ref=(\roman*)]
        \setcounter{enumi}{\value{enumCounterToolbox}}
        
        \item\label{cert is easier} (certification is easier than decision) $\mathsf{M}^{\cert}(f)\leq \mathsf{M}(f)$
        
        \setcounter{enumCounterToolbox}{\value{enumi}}
    \end{enumerate}
    %the following also hold
    \begin{enumerate}[label=(\roman*), ref=(\roman*)]
        \setcounter{enumi}{\value{enumCounterToolbox}}
        
        \item\label{cert idempotent}  (the certification transform is idempotent)
        $\of*{\mathsf{M}^{\cert}}^{\cert}(f)=\mathsf{M}^{\cert}(f)$
        
        \item\label{A > B > A cert = B cert} (Sandwich Lemma) If for all Boolean (possibly partial) functions $g$,
        $\mathsf{M}^{\cert}(g) \leq \mathsf{N}(g) \leq \mathsf{M}(g)$,
        then
        $$\mathsf{M}^{\cert}(f)=\mathsf{N}^{\cert}(f).$$
\end{enumerate}

There are a few non-stable measures that also verify items~(ii)--(iv) trivially because they collapse to a constant. The only measures in this paper (i.e. among all the measures defined in \Cref{sec:preliminaries}) for which items~(ii)--(iv) fail to hold are $\cc_1$, $\FC_1$, $\UC_1$ and $\ndeg$.

\end{restatable}

The Sandwich Lemma is key to proving the collapse of various certification measures.
In addition to LP duality, it provides further explanation to the surprising fact that there are so many equivalent definitions of certification measures such as $\RC$~\cite{AKP+21}.

Figure~\ref{fig:all} summarizes the main relations and separations between some complexity measures and their certification measures.

\newcommand{\newnode}[6]{\node [#6] (#1) at ({#2-3.3*#3},{(#4*4)+(17*#3/9)+(0.8*#2/3)}){#5}}
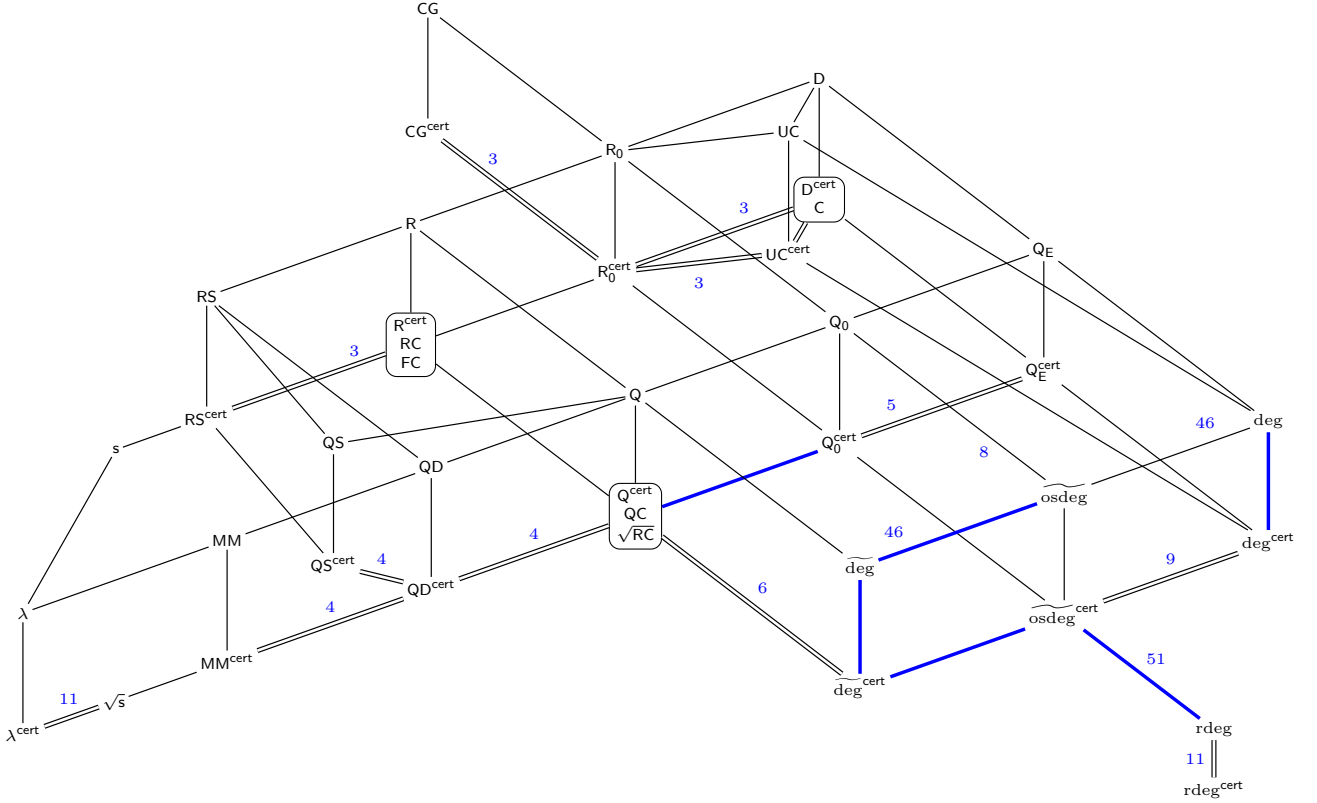
\begin{figure}[H]
\centering
\begin{tikzpicture}[
  x=0.3cm,
  y=0.4cm,
  every node/.style={font=\tiny, inner sep=1.2pt, fill=white},
  relation/.style={draw=black, line width=0.45pt},
  equality/.style={draw=black, double=none,
                   double distance=1pt, line width=0.45pt},
  separation/.style={draw=blue, line width=1.3pt}
]

%measures
\newnode{lambda}{0}{6}{1}{$\lambda$}{};
\newnode{lambdacert}{0}{6}{0}{$\lambda^\cert$}{};
%\draw[relation] (lambda)--(lambdacert);

\newnode{sqrts}{4}{6}{0}{$\sqrt{\s}$}{};
\newnode{s}{14}{9}{0}{$\s$}{};

%\node (s)       at (-9,6.43){$\s$};
%\node (bs)      at (-6,6.86){$\bs$};

\newnode{MM}{9}{6}{1}{$\MM$}{};
\newnode{MMcert}{9}{6}{0}{$\MM^\cert$}{};
%\draw[relation] (MM)--(MMcert);

\newnode{QS}{16}{6.7}{1}{$\QS$}{};
\newnode{QScert}{16}{6.7}{0}{$\QS^\cert$}{};
%\draw[relation] (QS)--(QScert);

\newnode{RS}{18}{9}{1}{$\RS$}{};
\newnode{RScert}{18}{9}{0}{$\RS^\cert$}{};
%\draw[relation] (RS)--(RScert);

\newnode{QD}{18}{6}{1}{$\QD$}{};
\newnode{QDcert}{18}{6}{0}{$\QD^\cert$}{};
%\draw[relation] (QD)--(QDcert);

\newnode{R}{27}{9}{1}{$\rqc$}{};
\newnode{Rcert}{27}{9}{0}{\begin{tabular}{c}
     $\rqc^\cert$\\ $\RC$\\ $\FC$
\end{tabular}}{draw,rectangle,rounded corners, align=center, inner xsep=-0.3em};
%\draw[relation] (R)--(Rcert);

\newnode{Q}{27}{6}{1}{$\qqc$}{};
\newnode{Qcert}{27}{6}{0}{\begin{tabular}{c}
     $\qqc^\cert$\\ $\QC$\\ $\sqrt{\RC}$
\end{tabular}}{draw,rectangle,rounded corners, align=center, inner xsep=-0.3em};
%\draw[relation] (Q)--(Qcert);

\newnode{adeg}{27}{3}{1}{$\adeg$}{};
\newnode{adegcert}{27}{3}{0}{$\adeg^\cert$}{};

\newnode{CG}{36}{11.5}{1}{$\CG$}{};
\newnode{CGcert}{36}{11.5}{0}{$\CG^\cert$}{};
%\draw[relation] (CG)--(CGcert);

\newnode{R0}{36}{9}{1}{$\roc$}{};
\newnode{R0cert}{36}{9}{0}{$\roc^\cert$}{};
%\draw[relation] (R0)--(R0cert);

\newnode{Q0}{36}{6}{1}{$\qoc$}{};
\newnode{Q0cert}{36}{6}{0}{$\qoc^\cert$}{};
%\draw[relation] (Q0)--(Q0cert);

\newnode{D}{45}{9}{1}{$\dqc$}{};
\newnode{Dcert}{45}{9}{0}{\begin{tabular}{c}
     $\dqc^\cert$\\ $\cc$
\end{tabular}}{draw,rectangle,rounded corners, align=center, inner xsep=-0.3em};
%\draw[relation] (D)--(Dcert);

\newnode{osdeg}{36}{3}{1}{$\osdeg$}{};
\newnode{osdegcert}{36}{3}{0}{$\osdeg^\cert$}{};
%\draw[relation] (osdeg)--(osdegcert);

\newnode{ndeg}{36}{1}{0}{$\rdeg$}{};
\newnode{ndegcert}{36}{1}{-0.5}{$\rdeg^\cert$}{};

\newnode{UC}{42}{8.5}{1}{$\UC$}{};
\newnode{UCcert}{42}{8.5}{0}{$\UC^\cert$}{};
%\draw[relation] (UC)--(UCcert);

%\newnode{C}{42}{9}{0}{$\cc$}{};

\newnode{QE}{45}{6}{1}{$\qec$}{};
\newnode{QEcert}{45}{6}{0}{$\qec^\cert$}{};
%\draw[relation] (QE)--(QEcert);

\newnode{deg}{45}{3}{1}{$\deg$}{};
\newnode{degcert}{45}{3}{0}{$\deg^\cert$}{};
%\draw[relation] (deg)--(degcert);

%equalities and relations before so that they are on the back
\draw[equality] (CGcert)--(R0cert) node[midway,above,,xshift=-10pt,yshift=12pt]{\ref{characterizations of C and RC}};
\draw[equality] (Dcert)--(UCcert);
\draw[equality] (UCcert)--(R0cert) node[midway,yshift=-7.5pt] {\ref{characterizations of C and RC}};
\draw[relation] (Q0cert)--(osdegcert);
\draw[relation] (R0cert)--(Q0cert);

%lines in the direction of approximation
\draw[relation] (D)--(R0);
\draw[relation] (R0)--(R);
\draw[relation] (R)--(RS);

\draw[equality] (Dcert)--(R0cert) node[midway,above,xshift=11pt,yshift=7.5pt] {\ref{characterizations of C and RC}};
\draw[relation] (R0cert)--(Rcert);
\draw[equality] (Rcert)--(RScert) node[midway,above,xshift=17pt,yshift=8pt]{\ref{characterizations of C and RC}};
\draw[relation] (RScert)--(s);

\draw[relation] (QE)--(Q0);
\draw[relation] (Q0)--(Q);
\draw[relation] (Q)--(QS);
\draw[relation] (Q)--(QD);
\draw[relation] (QD)--(MM);
\draw[relation] (MM)--(lambda);
\draw[equality] (QEcert)--(Q0cert) node[midway,above,xshift=-19pt,yshift=-2.5pt]{\ref{Q_0 cert = Q_E cert}};
\draw[separation] (Q0cert)--(Qcert);
\draw[equality] (Qcert)--(QDcert) node[midway,above,yshift=3.5pt]{\ref{characterizations of QC: computational}};
\draw[equality] (QDcert)--(MMcert) node[midway,above,yshift=3.5pt] {\ref{characterizations of QC: computational}}; 
\draw[relation] (MMcert)--(sqrts);
\draw[equality] (sqrts)--(lambdacert) node[midway,above,xshift=-1pt,yshift=3.5pt] {\ref{sanity check}};

\draw[relation] (deg)--(osdeg) node[midway,above,xshift=13,yshift=9pt] {\ref{adeg < osdeg < deg}};
\draw[separation] (osdeg)--(adeg) node[midway,above,xshift=-24,yshift=-3.7pt] {\ref{adeg < osdeg < deg}};

\draw[equality] (degcert)--(osdegcert) node[midway,above,yshift=3.5pt] {\ref{deg cert = osdeg cert}};
\draw[separation] (osdegcert)--(adegcert)  ;

%lines in the direction game - classical - quantum - polynomial
\draw[relation] (D)--(QE);
\draw[relation] (QE)--(deg);

\draw[relation] (Dcert)--(QEcert);
\draw[relation] (QEcert)--(degcert);

\draw[relation] (CG)--(R0);
\draw[relation] (R0)--(Q0);
\draw[relation] (Q0)--(osdeg) node[midway,xshift=13pt,yshift=-17pt]{\ref{osdeg < Q_0}};

%\draw[equality] (CGcert)--(R0cert);
%\draw[relation] (R0cert)--(Q0cert);
%\draw[relation] (Q0cert)--(osdegcert);
\draw[separation] (osdegcert)--(ndeg) node[midway,above,xshift=5.5pt,yshift=2pt]{\ref{ndeg < osdeg cert}};

\draw[relation] (R)--(Q);
\draw[relation] (Q)--(adeg);

\draw[relation] (Rcert)--(Qcert);
\draw[equality] (Qcert)--(adegcert) node[midway,above,,xshift=4pt,yshift=3pt]{\ref{characterizations of QC: polynomial}};

\draw[relation] (RS)--(QD);
\draw[relation] (RS)--(QS);

\draw[relation] (RScert)--(QScert);

%additional lines
\draw[relation] (D)--(UC);
\draw[relation] (UC)--(R0);
\draw[relation] (UC)--(deg);
%\draw[equality] (Dcert)--(UCcert);
%\draw[equality] (UCcert)--(R0cert);
\draw[relation] (UCcert)--(degcert);
%\draw[equality] (Dcert)--(C);
%\draw[equality] (Rcert)--(RC);
%\draw[equality] (Qcert)--(QC);
\draw[equality] (QDcert)--(QScert) node[midway,above,yshift=3.5pt]{\ref{characterizations of QC: computational}};
\draw[equality] (ndeg)--(ndegcert) node[midway,xshift=-7pt] {\ref{sanity check}};
%\draw[relation] (Rcert)--(bs);
%\draw[relation] (bs)--(s);
\draw[relation] (s)--(lambda);

%%lines in the direction of cert
\draw[relation] (lambda)--(lambdacert);
\draw[relation] (MM)--(MMcert);
\draw[relation] (RS)--(RScert);
\draw[relation] (QS)--(QScert);
\draw[relation] (QD)--(QDcert);
\draw[relation] (R)--(Rcert);
\draw[relation] (CG)--(CGcert);
\draw[relation] (Q)--(Qcert);
\draw[relation] (R0)--(R0cert);
\draw[relation] (UC)--(UCcert);
\draw[separation] (adeg)--(adegcert);
\draw[relation] (Q0)--(Q0cert);
\draw[relation] (D)--(Dcert);
\draw[relation] (osdeg)--(osdegcert);
\draw[relation] (QE)--(QEcert);
\draw[separation] (deg)--(degcert);

\end{tikzpicture}
\caption{ \label{fig:all}Relations between the complexity measures for all Boolean (possibly partial) functions $f$. A single line from $\mathsf{N}$ to $\mathsf{M}$ from down to up means
$\mathsf{N}(f) \leq O\of*{\mathsf{M}(f)}$.
Double lines, as well as measures being in the same box, denote equality up to $O(1)$ multiplicative factors, i.e. $\mathsf{N}(f) = \Theta\of*{\mathsf{M}(f)}$. Numbers indicate the corresponding theorems in which such relations are shown in this work. In addition, when the line is \textcolor{blue}{\textbf{thick and blue}}, it furthermore indicates a polynomial separation by total function in presented in this article, i.e. there exist a constant power
$k > 1$
and a family of Boolean total functions
$f_n : \zone^n \to \zone$
such that
${\mathsf{N}\of*{f_n}}^k \leq O\of*{\mathsf{M}\of*{f_n}}$. These separations witnessing that some complexity measures are not equivalent follow from \Cref{thm:separation (deg cert)^2.5 < ndeg}, \Cref{thm:Q_0 algo} and \Cref{separation: deg^2 < deg cert and adeg^2 < adeg cert}.
The definitions of the complexity measures can be found in \Cref{sec:preliminaries}.}
\label{fig:complexity-relations}
\end{figure}

\subsection{Related work}

Certificate complexity was defined by Nisan~\cite{Nis89}, and has been widely studied since. Randomized and quantum analogues ($\RC$ and $\QC$) were proposed by Aaronson~\cite{Aar08}. Aaronson used a similar operational definition of certificates that we study in this paper.
Tal~\cite{Tal13} defined an equivalent measure $\RC$ for the randomized setting, formulated as a linear program.
 In later work, \cite{BK19} gave equivalent formulations as the complexity of the partial function $f^\cert_z$. The present work studies this formulation systematically, applying it in a variety of contexts.

A witness in a nondeterministic algorithm could be considered to be a natural candidate for 
a certificate object. 
The exact postselected randomized query complexity indeed coincides with classical certificate complexity~\cite[Theorem~16]{Cade20}. However, in the quantum setting, it is characterized by rational degree (which is equivalent to the maximum between the nondeterministic degree of $f$ and of~$\neg f$). 
This suggests that rational degree can be viewed as another natural quantum analogue of certificate complexity. However, this notion is distinct from Aaronson's quantum certificate complexity~\cite{Aar08}.

Taking a different approach, Jeffery~\cite{Jeffery2026QuantumQuery} proposes a notion of quantum witnesses, which can be viewed as witnesses for span programs, known to provide a characterization of quantum query complexity.  The corresponding complexity measure coincides with $\qqc$, suggesting that certificate objects and certification complexity do not coincide for  quantum query complexity.

The presence of these different, non-equivalent candidates highlights a more fundamental question: what is the appropriate notion of a certificate in the quantum query model, and which aspects of classical certification should such a notion preserve? 

Returning to the classical setting, expectational certificate complexity ($\EC$, \Cref{def:EC})~\cite{JKK+20} defines an expectational version of randomized certificates. $\EC$ certificates can be used to give a certification algorithm, however, it is not known to be equivalent to $\RC$.

\section{Preliminaries}
\label{sec:preliminaries}

We provide all the necessary definitions in this section.
We refer the reader to \cite{BdW02} and \cite{ABBL+17} for more formal definitions of the deterministic query model of computation, the randomized query model of computation and the quantum query model of computation for Boolean functions.

\subsection{Query measures}
\begin{defi}[deterministic query complexity]\label{def:dqc}
    The deterministic query complexity of a Boolean (possibly partial) function $f$, denoted by $\dqc(f)$, is the worst-case cost of the best deterministic query algorithm computing $f$.
\end{defi}

\begin{defi}[randomized query complexity]
\label{def:rqc}
    The randomized query complexity of a Boolean (possibly partial) function $f$, denoted by $\rqc(f)$, is the worst-case cost of the best randomized query algorithm computing $f$ with two-sided error at most $1/3$.
\end{defi}

\begin{defi}[zero-error randomized query complexity]
\label{def:roc}
    The zero-error randomized query complexity of a Boolean (possibly partial) function
    $f : \D \subseteq \zone^n \to \zone$,
    denoted by $\roc(f)$, is the minimum worst-case expected cost of a randomized query algorithm that computes $f$ with zero-error, i.e. on every input
    $x \in \D$
    the algorithm should give the correct answer $f(x)$ with probability $1$. 
\end{defi}

\begin{defi}[quantum query complexity]\label{def:qqc}
    Let
    $f : \D \subseteq \zone^n \to \zone$
    be a Boolean (possibly partial) function. The \emph{bounded-error quantum query complexity} of $f$, denoted by $\qqc(f)$, is the minimum number of queries made by a quantum query algorithm $A$ that, for every input $x\in\D$, outputs $f(x)$ with probability at least $\frac{2}{3}$, i.e.
    $$\Pr\bracket*{A^x=f(x)} \geq \frac{2}{3}$$
    where $A^x$ denotes $A$ having query access to input $x$.
\end{defi}

The constant $2/3$ in the definition is arbitrary: it can be replaced
by any constant strictly larger than $1/2$, with only a constant-factor
change in the query complexity, by standard amplification.

\begin{defi}[zero-error quantum query complexity]\label{def:Q0}
    Let
    $f : \D \subseteq \zone^n \to \zone$
    be a Boolean (possibly partial) function. The \emph{zero-error quantum query complexity} of $f$, denoted by $\qoc(f)$, is the number of queries made by a quantum query algorithm that on any input $x \in \D$ is allowed to output $\bot$ with probability at most $1/2$. When it outputs a Boolean value $b$, $f(x)=b$ with probability $1$.
\end{defi}

\begin{defi}[exact quantum query complexity]\label{def:qec}
    The exact quantum query complexity of a Boolean (possibly partial) function
    $f : \D \subseteq \zone^n \to \zone$,
    denoted by $\qec(f)$, is the minimum number of queries made by a quantum algorithm that outputs $f(x)$ on every input
    $x \in \D$
    with probability $1$.
\end{defi}

\subsection{Certificates and related measures}

\begin{defi}[partial assignment] \label{def:partial assignment}
    A word
    $\rho \in \set*{0,1,*}^n$
    is said to be a \emph{partial assignment over $\zone^n$} (where $*$ stands for ``not assigned''). If $\rho$ and $\tau$ are partial assignments over $\set*{0,1}^n$, then $\tau$ is said to be \emph{consistent} with $\rho$, denoted
    $\rho \prec \tau$,
    if
    $$\for{i\in[n]} \rho(i) \neq * \Ra \rho(i) = \tau(i).$$
    The \emph{number of assigned positions in $\rho$} is
    $$\card*{\rho} := \card*{\rho^{-1}\of*{\zone}}.$$
    Partial assignments $\rho$ and $\tau$ are said to \emph{conflict} if there exists an index
    $i \in [n]$
    where
    $* \neq \rho(i) \neq \tau(i) \neq *$.
\end{defi}

\begin{defi}[certificate] \label{def:certificate}
    Let
    $f : \D \subseteq \zone^n \to \zone$
    be a Boolean (possibly partial) function.
    For all inputs
    $x \in \D$,
    a partial assignment
    $c\in \set*{0,1,*}^n$
    is said to be a \emph{certificate for $x$} if
    $c \prec x$
    and for all
    $z\in\D$
    such that
    $c \prec z$
    we have
    $f(x)=f(z)$.
\end{defi}

Note that partial functions admit two extensions of the definition of certificate for total functions. We use the weak version, where a certificate for $x$ should only conflict with inputs
$y \in \D \setminus f^{-1}(f(x))$,
whereas in the strong version it should conflict with all elements of
$y \in \zone^n \setminus f^{-1}(f(x))$
which includes the elements outside of the domain.

\begin{defi}[certificate complexity] \label{def:C}
    Let
    $f : \D \subseteq \zone^n \to \zone$
    be a Boolean (possibly partial) function.
    For all input
    $x \in \D$
    the \emph{certificate complexity of $f$ at $x$}, denoted by $\cc(f,x)$, is defined as
    $$\cc(f,x):=\min\limits_{c\predicate{ certificate for }x} \card*{c}.$$
    For
    $b \in \zone$
    the \emph{$b$-certificate complexity of $f$} is defined as
    $$\cc_b(f):=\max\limits_{x\in f^{-1}(b)} \cc(f,x).$$
    Then the \emph{certificate complexity of $f$}, denoted $\cc(f)$, is defined as
    $$\cc(f):=\max\limits_{x\in\D} \cc(f,x) = \max\set*{\cc_0(f),\cc_1(f)}.$$

\end{defi}

\begin{defi}[unambiguous certificate complexity] \label{def:UC}
    Let
    $b \in \zone$.
    Let
    $f : \D \subseteq \zone^n \to \zone$
    be a Boolean (possibly partial) function.
    Let a family of partial assignments
    $\of*{c_x}_{x\in\D}$. The family $\C =\of*{c_x}$ is said to be an \emph{unambiguous family of certificates w.r.t. $f$} if for all
    $x,y \in \zone^n$
    \begin{itemize}
        \item certificate condition: $c_x$ is a certificate for $x$,
        \item unambiguity condition: either
        $c_x = c_y$,
        or $c_x$ and $c_y$ conflict.
    \end{itemize}
    For
    $b \in \zone$
    the \emph{unambiguous $b$-certificate complexity of $f$} is defined as
    $$\UC_b(f):= \limitB{\min}{\of*{c_x}\ \predicate{unambiguous family}}{\predicate{of certificates w.r.t. }f}\ \max\limits_{x \in f^{-1}(b)} \card*{c_x}.$$
    The \emph{unambiguous certificate complexity of $f$}, denoted $\UC(f)$, is defined as
    $$\UC(f) := \max\set*{\UC_0(f),\UC_1(f)}.$$

\end{defi}

Note that removing the unambiguity condition gives $\cc_0$, $\cc_1$ and $\cc$.% and $\cc_\text{min}$.

\begin{defi}[fractional certificate complexity $\FC$]\label{def:FC}
    Let
    $f : \D \subseteq \zone^n \to \zone$
    be a Boolean (possibly partial) function. For
    $b \in \zone$,
    the \emph{fractional $b$-certificate complexity of $f$}, denoted $\FC_b(f)$, is the value of the following linear program.
    \begin{align*}
        \FC_b(f) \qquad
        &\predicate{is the }\min\predicate{ of}\quad \max\limits_{x \in f^{-1}(b)} \sum\limits_{i \in [n]} w_x(i)\\
        &\predicate{over}\quad \of*{w_x}_{x \in \D}\\
        &\predicate{subject to}\quad \for{x \in \D} \for{i \in [n]} 0 \leq w_x(i)\\
        &\predicate{and the weak conflict condition}\\
        &\for{x \in \D} \for{y \in \D}\\
        &f(x) \neq f(y) \Ra
        \sum\limits_{i \in [n]: x_i \neq y_i} w_x(i) \geq 1 
    \end{align*}
    Weights $\of*{w_x}_{x \in \D}$ subject to such conditions are said to be a \emph{family of fractional $b$-certificates w.r.t. $f$}.
    
    The \emph{fractional certificate complexity of $f$}, denoted $\FC(f)$, is defined to be $\max\{\FC_0(f),\FC_1(f)\}$. 
    
\end{defi}

Because each fractional certificate of a family of fractional certificates satisfies the conditions independently of the rest of the family, and similarly for the fractional $b$-certificates, then
$\FC(f) = \max\set*{\FC_0(f),\FC_1(f)}$.

\begin{defi}[expectational certificate complexity $\EC$ {\cite[Definition~9]{JKK+20}}]\label{def:EC}
    Let
    $f : \D \subseteq \zone^n \to \zone$
    be a Boolean (possibly partial) function.
    The \emph{expectational certificate complexity of $f$}, denoted $\EC(f)$, is the value of the following linear program.
    \begin{align*}
        \EC(f) \qquad
        &\predicate{is the }\min\predicate{ of}\quad \max\limits_{x \in \D} \sum\limits_{i \in [n]} w_x(i)\\
        &\predicate{over}\quad \of*{w_x}_{x \in \D}\\
        &\predicate{subject to}\quad \for{x \in \D} \for{i \in [n]} 0 \leq w_x(i) \leq 1\\
        &\predicate{and the strong conflict condition}\\
        &\for{x \in \D} \for{y \in \D}\\
        &f(x) \neq f(y) \Ra
        \sum\limits_{i \in [n]: x_i \neq y_i} w_x(i) w_y(i) \geq 1 
    \end{align*}
    Weights $\of*{w_x}_{x \in \D}$ subject to such conditions are said to be a {\emph{family of expectational certificates w.r.t. $f$}}.
\end{defi}

\begin{defi}[certificate game complexity with private randomness $\CG$~\cite{CGL+23}] \label{def:CGpriv}
    Let
    $f : \D \subseteq \zone^n \to \zone$
    be a Boolean (possibly partial) function. The \emph{certificate game complexity with private randomness of $f$}, denoted $\CG(f)$, is given by the following linear program.
    \begin{align*}
        \sqrt{\CG(f)} \qquad
        &\predicate{is the }\min\predicate{ of}\quad \max\limits_{x \in \D} \sum\limits_{i \in [n]} w_x(i)\\
        &\predicate{over}\quad \of*{w_x}_{x \in \D}\\
        &\predicate{subject to}\quad \for{x \in \D} \for{i \in [n]} 0 \leq w_x(i)\\
        &\predicate{and the strong conflict condition}\\
        &\for{x \in \D} \for{y \in \D}\\
        &f(x) \neq f(y) \Ra
        \sum\limits_{i \in [n]: x_i \neq y_i} w_x(i) w_y(i) \geq 1 
    \end{align*}
\end{defi}

\subsection{Adversary bounds and other measures}

\begin{defi}[min max adversary bound $\MM$] \label{def:MM}
    Let
    $f : \D \subseteq \zone^n \to \zone$
    be a Boolean (possibly partial) function.
    The \emph{min max adversary bound of $f$}, denoted $\MM(f)$, is the value of the following program.
    \begin{align*}
        \MM(f) \qquad
        &\predicate{is the }\min \predicate{ of} \quad \max\limits_{x \in \D} \sum_{i \in [n]} w_x(i)\\
        &\predicate{over} \quad \of*{{w_{x}(i)}}_{x \in \D, i \in [n]}\\
        &\predicate{subject to}\quad \for{x \in \D} \for{i \in [n]} w_x(i) \geq 0\\
        &\predicate{and}\\
        &\for{x \in \D} \for{y \in {\D}}\\
        &f(x) \neq f(y) \Ra
        \sum_{i \in [n]: x_i \neq y_i} \sqrt{w_x(i) w_y(i)} \geq 1 
    \end{align*}
\end{defi}

\begin{defi}[classical adversary bound $\CAdv$~{\cites[where $\CAdv$ is denoted $\mathrm{CMM}$, see Lemma~7 for the equivalence with this definition]{AKP+21}[Section~2.2.2]{ABK21}}]\label{def:CAdv}
    Let
    $f : \D \subseteq \zone^n \to \zone$
    be a Boolean (possibly partial) function.
    The \emph{classical adversary bound of $f$}, denoted $\CAdv(f)$, is the value of the following program.
    \begin{align*}
        \CAdv(f) \qquad
        &\predicate{is the }\min \predicate{ of} \quad \max\limits_{x \in \D} \sum_{i \in [n]} w_x(i)\\
        &\predicate{over} \quad \of*{{w_{x}(i)}}_{x \in \D, i\in [n]}\\
        &\predicate{subject to}\quad \for{x \in \D} \for{i \in [n]} w_x(i) \geq 0\\
        &\predicate{and}\\
        &\for{x \in \D} \for{y \in {\D}}\\
        &f(x) \neq f(y) \Ra
        \sum_{i \in [n]: x_i \neq y_i} \min\set*{w_x(i),w_y(i)} \geq 1 
    \end{align*}
\end{defi}

We give the primal definition of the partition bound for query complexity.

\begin{defi}[partition bound {\cite[Definition~13]{JK10}}]\label{def:prt}
    Let
    $f:\D \subseteq \zone^n\to\zone $
    be a Boolean (possibly partial) function.
    Let $\A$ denote the set of all partial assignments
    $A : S \to \zone$,
    where
    $S \subseteq [n]$.
    We write
    $x \in A$
    if
    $x_i = A(i)$
    for every
    $i \in S$,
    and define
    $\card*{A}:=\card*{S}$.

    For
    $\varepsilon \geq 0$,
    the $\varepsilon$-partition bound of $f$, denoted by $\Prt_\varepsilon(f)$, is the optimal value of the following linear program:

    \begin{align*}
        \Prt_\varepsilon(f) \qquad
        &\predicate{is the }\min\predicate{ of}&& \sum\limits_{b\in\zone}
        \sum\limits_{A\in\A}
        w_{b,A}\,2^{\card*{A}}\\
        &\predicate{over weights}&& \of*{w_{b,A}}_{b\in\zone,A\in\A}\\
        &\predicate{subject to}&& \sum\limits_{\substack{A\in\A\\x\in A}} w_{f(x),A}
        \geq 1-\varepsilon
        && \for{x\in\D}\\
        &\predicate{and}&&\quad \sum\limits_{\substack{A\in\A\\x\in A}}
        \sum\limits_{b\in\zone} w_{b,A} = 1
        && \for{x\in\zone^n}\\
        &\predicate{and}&& 
        w_{b,A}\geq 0
        && \for{b\in\zone} \for{A\in\A}
     \end{align*}

    The partition bound is scaled exponentially compared to other query lower bounds, so we define
    $\prt_\varepsilon(f) := \log_2\of*{\Prt_\varepsilon(f)}$.
    Also define
    $\prt(f) := \prt_{\frac{1}{3}}(f)$.
\end{defi}

\begin{defi}[sabotage complexity]\label{def:sabotage complexity}
    Let
    $f : \D \subseteq \zone^n \to \zone$
    be a Boolean (possibly partial) function. For
    $b \in \set*{\ast,\dagger}$,
    define
    $$P_f^b := \comprehension{\rho \in \set*{0,1,b}^n}
    {\predicate{there exist }x\in f^{-1}(0)\predicate{ and }y\in f^{-1}(1)\predicate{ such that }\rho \prec x\predicate{ and }\rho \prec y}.$$
    The \emph{sabotage function} associated to $f$ is 
    $f_\sab : P_f^\ast\cup P_f^\dagger \to \zone$,
    defined by
    $$f_\sab(z):=
    \begin{cases}
        0 &\predicate{if }z\in P_f^\ast,\\
        1 &\predicate{if }z\in P_f^\dagger
    \end{cases}$$
The \emph{randomized sabotage complexity of $f$} is defined as
$$\RS(f) := \roc\of*{f_\sab}$$
where $\roc$ denotes zero-error randomized query complexity.
\end{defi}

\subsection{Sensitivity and related measures}

\begin{defi}[sensitivity] \label{def:s}
    Let
    $f : \D \subseteq \zone^n \to \zone$
    be a Boolean (possibly partial) function.
    For all input
    $x \in \D$
    the $i^{\text{th}}$ bit is said to be \emph{sensitive for $x$ w.r.t. $f$} if
    $f\of*{x^i} \neq f(x)$
    where $x^i$ denotes the input equal to $x$ on every bit except on bit $i$, which is flipped.
    The \emph{sensitivity of $x$ w.r.t $f$} is defined as
    $$\s(f,x) := \card*{\comprehension{i \in [n]}{f\of*{x^i} \neq f(x),x^i \in \D}},$$
    while the \emph{sensitivity of $f$} is defined as
    $$\s(f) := \max\limits_{x\in\D} \s(f,x).$$
\end{defi}

Spectral sensitivity was first introduced by \cite{Kou93} to prove formula size lower bounds.  It is a spectral relaxation of sensitivity.
\begin{defi}\label{def:lambda}
Let
    $f : \D \subseteq \zone^n \to \zone$
    be a Boolean (possibly partial) function.
    We define $M_f$ as the $\D\times\D$ sensitivity matrix of $f$ as $M_f(x,y) = 1$ whenever $f(x)\neq f(y)$ and $ d_H(x,y)=1$, and 0 otherwise.
    Its spectral sensitivity is defined as 
    $$ \lambda(f) = \norm*{M_f}$$
    where $\norm{\cdot}$ is the spectral norm.
\end{defi}

\begin{defi}[block sensitivity] \label{def:bs}
    Let
    $f : \D \subseteq \zone^n \to \zone$
    be a Boolean (possibly partial) function.
    For all input
    $x \in \D$
    a block
    $B\subseteq[n]$
    is said to be a \emph{sensitive block for $x$ w.r.t $f$} if $x^B\in \D$ and
    $f(x^B)\neq f(x)$, %\sophie{!!!!!!!!!!!  what does f(x) neq f(y) mean when f(x) or f(y) is undefined.}
    where $x^B$ denotes the input equal to $x$ on every bit except on bits
    $i \in B$,
    which are flipped.
    The \emph{block sensitivity of $f$ at $x$}, denoted by $\bs(f,x)$, is defined as 
    $$\bs(f,x):=\max\comprehension{k}{\exists\:\! B_{1}, \dots ,B_{k}\predicate{ with }B_{i}\cap B_{j}=\emptyset\predicate{ for }i\neq j\predicate{ and }f\of*{x^{B_{i}}}\neq f(x)\predicate{ and }x^{B_i} \in \D}.$$
    The \emph{block sensitivity of $f$} is defined as
    $$\bs(f):=\max\limits_{x\in\D} \bs(f,x).$$
\end{defi}

Fractional block sensitivity ($\fbs$) is obtained by allowing fractional weights on sensitive blocks.

\begin{defi}[fractional block sensitivity] \label{def:fbs}
    Let
    $f : \D \subseteq \zone^n \to \zone$
    be a Boolean (possibly partial) function.
    For all input
    $x \in \D$
    let
    $\mathcal{B}(f,x):=\comprehension{B\subseteq [n]}{f\of*{x^{B}}\neq f(x),x^B\in\D}$
    denote the set of all its sensitive blocks.
    The \emph{fractional block sensitivity of $f$ at $x$}, denoted by $\fbs(f,x)$, is the value of the following linear program.
    \begin{align*}
        \fbs(f,x) \qquad
        &\predicate{is the }\max\predicate{ of}\quad \sum_{B\in \mathcal{B}(f,x)}w(B)\\
        &\predicate{over}\quad \of*{w(B)}_{B\in \mathcal{B}(f,x)}\\
        &\predicate{subject to}\quad \for{i\in[n]} \sum_{B\in \mathcal{B}(f,x):i\in B}w(B)\leq 1\\
        &\predicate{and}\quad \for{B\in \mathcal{B}(f,x)} w(B)\in \bracket*{0,1}
     \end{align*}
    The \emph{fractional block sensitivity of $f$} is defined as:
    $$\fbs(f):=\max\limits_{x\in\D} \fbs(f,x).$$
\end{defi}

Note that restricting the linear program for $\fbs(f,x)$ to  integral values gives $\bs(f,x)$.

The fractional measures $\fbs$ and $\FC$ were introduced in \cite{Tal13}, where it was shown that for all
$x\in\zone^{n}$
they verify
$\fbs(f,x)=\FC(f,x)$
since the linear program for fractional certificate complexity and fractional block sensitivity are the primal-dual of each other and are also feasible.

Next we define critical fractional block sensitivity from fractional block sensitivity.

\begin{defi}[critical fractional block sensitivity~{\cite[Section~2.2.1]{ABK21}}]\label{def:cfbs}
    Let
    $f : \D(f) \subseteq \zone^n \to \zone$
    be a Boolean (possibly partial) function. A Boolean total function
    $g : \zone^n \to \zone$
    is called a \emph{completion of $f$} if for all inputs
    $x \in \D$
    they coincide, i.e.
    $g(x)=f(x)$.
    The \emph{critical fractional block sensitivity of $f$}, denoted $\cfbs(f)$, is defined as
    $$\cfbs(f):=\min\limits_{g\predicate{ completion of }f}\ \max\limits_{x\in\D(f)} \fbs\of*{x,g}.$$
\end{defi}

\subsection{Polynomial degree measures}\label{sec:def deg measures}

Degree and approximate degree are defined in Section~\ref{subsec:polynomial-certification}. We give the definitions of a few more polynomial degree-based measures in this section.

We start with the two only degree measures where we do not require the polynomials to be bounded on the whole Boolean hypercube $\zone^n$.

\begin{defi}[nondeterministic polynomial degree]\label{def:ndeg}
    Let
    $ f:\D\subseteq \set*{0,1}^n \to \zone$
    be a Boolean (possibly partial) function.
    The \emph{nondeterministic polynomial degree of $f$}, denoted $\ndeg(f)$, is defined as
    \begin{align*}
        \ndeg(f) \qquad
        &\predicate{is the }\min\predicate{ of}\quad \deg(P)\\
        &\predicate{over real polynomial }P\\
        &\predicate{subject to}\quad \for{x \in f^{-1}(0)} P(x) = 0\\
        &\predicate{and}\quad \for{x \in f^{-1}(1)} P(x) \neq 0
    \end{align*}
\end{defi}

\begin{defi}[rational degree]\label{def:rdeg}
    Let
    $f:\D\subseteq \set*{0,1}^n \to \zone$
    be a Boolean (possibly partial) function.
    The \emph{rational polynomial degree of $f$}, denoted $\rdeg(f)$, is defined as
    \begin{align*}
        \rdeg(f) \qquad
        &\predicate{is the }\min\predicate{ of}\quad \max\set*{\deg(P),\deg(Q)}\\
        &\predicate{over real polynomials }P,Q\\
        &\predicate{subject to}\quad \for{x \in \D} \frac{P(x)}{Q(x)} = f(x)
    \end{align*}
\end{defi}

For all Boolean (possibly partial) functions $f$,
$\rdeg(f) = \max\set*{\ndeg(f),\ndeg(\neg f)}$.
We  use this characterization from \cite[Fact~3]{KKW+26} when working with $\rdeg$.

\begin{defi}[block-multilinear bounded-polynomial degree~{\cites{AA15}[Definition~2.3]{AAI+16}}]\label{def:bmdeg}
    Let
    $ f:\D\subseteq \set*{-1,1}^n \to \zone$ 
    be a Boolean (possibly partial) function. Let
    $\epsilon \geq 0$.
    A polynomial $P$ of degree $d$ is called \emph{block-multilinear} if its variables can be partitioned into $d$ blocks such that every monomial of $P$ contains exactly one variable from each block.

    For $x=(x_1,\ldots,x_n)\in\{-1,1\}^n$, let
    $$\widetilde{x}:=(1,x_1,\ldots,x_n).$$
    The \emph{block-multilinear $\epsilon$-approximate bounded-polynomial degree of $f$}, denoted by $\bmadeg[_\epsilon](f)$, is defined as the value of the following program.
    
    \begin{align*}
        \bmadeg[_\epsilon](f) \qquad
        &\predicate{is the }\min\predicate{ of}\quad d\\
        &\predicate{over}\quad d \in \NN\\
        &\predicate{and block-multilinear real polynomial }P:\mathbb{R}^{d(n+1)}\to\mathbb{R}\predicate{ of degree }d\\
        &\predicate{subject to}\quad \for{x \in \D} 0 \leq P(\widetilde{x},\ldots,\widetilde{x}) \leq 1\\
        &\predicate{and the error condition}\quad \for{x \in \D} \card*{P(\widetilde{x},\ldots,\widetilde{x})-f(x)} \leq \epsilon\\
        &\predicate{and the bounding condition}\quad \for{z\in\set*{-1,1}^{k(n+1)}} \card*{P(z)} \leq 1
     \end{align*}
    Then define the \emph{block-multilinear approximate bounded-polynomial degree of $f$}, denoted $\bmadeg(f)$,
    and the \emph{block-multilinear bounded-polynomial degree of $f$}, denoted $\bmdeg(f)$, as
    $$\bmadeg(f) := \bmadeg[_\frac{1}{3}](f) \qquad\qquad \bmdeg(f):=\bmadeg[_0](f).$$
\end{defi}

Given a block-multilinear polynomial $P$ approximating
$g : \D \subseteq \set*{-1,1}^n \to \zone$
in the $\set*{-1,1}$ basis, the polynomial $P \circ \frac{1+x}{2}$ is now in the $\zone$ basis with the same degree and same values. Similarly for $f : \D' \subseteq \zone^n \to \zone$ in the $\zone$ basis, the reverse transformation $f \circ \of*{2x-1}$ brings it to the $\set*{-1,1}$ basis. That's this relabeling that we have in mind when applying $\bmadeg[_\epsilon]$ to a function
$f$ in the $\zone$ basis~\cite[Section~2.2]{AAI+16}: first we go to $g := f \circ \of*{2x-1}$ and then apply $\bmadeg[_\epsilon]$ to $g$.

\section{Certification complexity}

Informally, an 
algorithm that certifies the value $f(z)$ of $z$ should output 1 on $z$ and 0 whenever it is given an input $x$ such that $f(x)\neq f(z)$.
The formal definition is given in 
\Cref{def:f cert} and \Cref{def:cert qc certification version}, where we have used the notation 
$$\mathsf{M}^\cert(f) := \max\limits_{z \in \D} \mathsf{M}\of*{f^\cert_z},$$
for $\mathsf{M}$ a complexity measure of Boolean functions.

\subsection{General properties of certification measures }\label{sec:general properties}
We prove several properties of certification complexity  that will be essential tools in proving the main results of the paper. We show that certification is non-increasing,
idempotent, and monotone with respect to the underlying query complexity
measure. We also give a general condition under which two measures
have the same certification complexity. These properties will be used
throughout the paper.

\Toolbox*

Figure~\ref{fig:toolbox} summarizes the properties in \Cref{toolbox}.

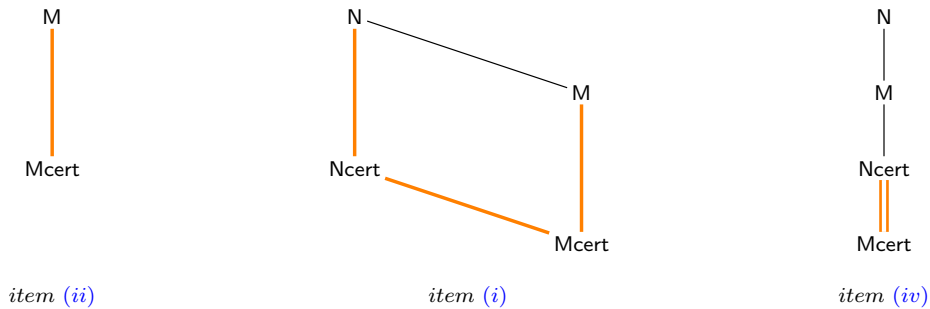
\begin{figure}[H]
\centering
\begin{tikzpicture}[
  x=1.cm,
  y=0.5cm,
  every node/.style={font=\scriptsize, inner sep=1.5pt, fill=white},
  relation/.style={draw=black, line width=0.45pt},
  implied/.style={draw=orange, line width=1.3pt},
  equality/.style={draw=black, double=white,
                   double distance=1.15pt, line width=0.65pt},
  impliedEq/.style={draw=orange, double=white,
                   double distance=1.6pt, line width=1pt}
]

\node (i)      at (1,-1.4){$\cref{cert is easier}$};
\node (iM)      at (1,6){$\mathsf{M}$};
\node (iMcert)  at (1,2){$\mathsf{M}\cert$};

\draw[implied] (iM)--(iMcert);

\node (ii)      at (6.5,-1.4){$\cref{cert monotonicity}$};
\node (iiM)      at (5,6){$\mathsf{N}$};
\node (iiMcert)  at (5,2){$\mathsf{N}\cert$};
\node (iiN)      at (8,4){$\mathsf{M}$};
\node (iiNcert)  at (8,0){$\mathsf{M}\cert$};
\draw[implied] (iiM)--(iiMcert);
\draw[implied] (iiN)--(iiNcert);
\draw[relation] (iiN)--(iiM);
\draw[implied] (iiNcert)--(iiMcert);

\node (iii)      at (12,-1.4){$\cref{A > B > A cert = B cert}$};
\node (iiiM)      at (12,6){$\mathsf{N}$};
\node (iiiMcert)  at (12,2){$\mathsf{N}\cert$};
\node (iiiN)      at (12,4){$\mathsf{M}$};
\node (iiiNcert)  at (12,0){$\mathsf{M}\cert$};
\draw[relation] (iiiM)--(iiiN);
\draw[relation] (iiiN)--(iiiMcert);
\draw[impliedEq] (iiiNcert)--(iiiMcert);

\end{tikzpicture}
\caption{Three items from the Toolbox (\Cref{toolbox}). Double lines denote equality up to $O(1)$ multiplicative factor. A single line from $\mathsf{N}$ to $\mathsf{M}$ from down to up means
$\mathsf{N}(f) \leq O\of*{\mathsf{M}(f)}$.
Double lines denote equality up to $O(1)$ multiplicative factors, i.e. $\mathsf{N}(f) = \Theta\of*{\mathsf{M}(f)}$. \textbf{\textcolor{orange}{Orange lines}} are the lines implied by black lines. The case of stable complexity measures is depicted.}
\label{fig:toolbox}
\end{figure}

\begin{proof}[Proof of \cref{cert monotonicity} from \Cref{toolbox}.]
    % Your proof of (iii)
    Let two complexity measures $\mathsf{M}$ and $\mathsf{N}$ such that
    $\mathsf{M}(g)\leq\mathsf{N}(g)$ for all Boolean (possibly partial) function~$g$. Applying the inequality to the $z$-certification versions of $f$, one gets
    $$\mathsf{M}^\cert(f) = \max\limits_{z\in\D} \mathsf{M}\of*{f^\cert_z} \leq \max\limits_{z\in\D} \mathsf{N}\of*{f^\cert_z} = \mathsf{N}^\cert(f).\qedhere$$
\end{proof}

\begin{proof}[Proof of \cref{cert is easier} of \Cref{toolbox}.]
    We prove that all complexity measures considered in this work, except $\cc_1$, $\FC_1$, $\UC_1$, and $\ndeg$, verify \cref{cert is easier} of \Cref{toolbox}. (See definitions in \Cref{sec:preliminaries} for the list of these measures.) Let $\mathsf{M}$ be any of these measures. Then $\mathsf{M}$ is monotone under restriction of the domain (\cref{item stable:monotonicity under restriction} of the definition of stability).

    Let first consider measures $\mathsf{M}$ that are stable (i.e. that also verify \cref{item stable:negation} of the definition). For all
    $z \in f^{-1}(1)$,
    $$\mathsf{M}\of*{f^\cert_z} = \mathsf{M}\of*{f_{\big|\set{z} \cup \comprehension{y\in\D}{f(z) \neq f(y)}}} \leq \mathsf{M}(f)$$
    and for all
    $z \in f^{-1}(0)$,
    $$\mathsf{M}\of*{f^\cert_z} = \mathsf{M}\of*{\neg f_{\big|\set{z} \cup \comprehension{y\in\D}{f(z) \neq f(y)}}} \leq \mathsf{M}(f).$$
    So
    $\mathsf{M}^\cert(f) = \max\limits_{z\in\D} \of*{f^\cert_z} \leq \mathsf{M}$.

    Now consider $\mathsf{M}(f)$ being one of $\ndeg(\neg f)$, $\cc_0(f)$ or $\FC_0(f)$ -- for which \cref{item stable:negation} of the definition fails. The result will still hold as these measures vanish through the $\cert$ transform, as follows.
    
    For any constant function
    $f : \D \subseteq \zone \to \zone$
    we have
    $\mathsf{M}(f) = 0$,
    so
    $$\mathsf{M}^\cert(f) = \max\limits_{z \in \D} \mathsf{M}\of*{f^\cert_z} = 0 = \mathsf{M}(f)$$
    where the second equality follows from $f^\cert_z$ being constant.
    Now, for any non constant $f$ and any input
    $z \in \D$,
    we have
    $\mathsf{M}(f) \leq 1$
    and
    $\mathsf{M}\of*{f^\cert_z} = 1$
    so
    $$\mathsf{M}^\cert(f) = \max\limits_{z \in \D} \mathsf{M}\of*{f^\cert_z} = 1 \leq \mathsf{M}(f).$$

\end{proof}

\begin{proof}[Proof of \cref{cert idempotent} from \Cref{toolbox}.]
For all
    $z \in \D$
    we have that
    $\of*{f^\cert_z}^\cert_z = f^\cert_z$.
    Then on the one hand,
    \begin{align*}
        \of*{\mathsf{M}^\cert}^\cert(f)
        &= \max\limits_{z \in \D} \mathsf{M}^\cert\of*{f^\cert_z}\\
       &= \max\limits_{z \in \D} \max\limits_{y \in \D} \mathsf{M}\of*{\of*{f^\cert_z}^\cert_y}\\
        &\geq \max\limits_{z \in \D} \mathsf{M}\of*{\of*{f^\cert_z}^\cert_z}\\
        &= \max\limits_{z \in \D} \mathsf{M}\of*{f^\cert_z}\\
        &= \mathsf{M}^\cert\of*{f}
    \end{align*}
    On the other hand,
    $$\of*{\mathsf{M}^\cert}^\cert(f)
    = \max\limits_{z \in \D} \mathsf{M}^\cert\of*{f^\cert_z}
    \leq \max\limits_{z \in \D} \mathsf{M}\of*{f^\cert_z}
    = \mathsf{M}^\cert(f)$$
    where the inequality follows from the assumption \cref{cert is easier} from \Cref{toolbox}.
\end{proof}

\begin{proof}[Proof of \cref{A > B > A cert = B cert} from \Cref{toolbox}.]
Suppose that
$\mathsf{M} \geq \mathsf{N} \geq \mathsf{M}^{\cert}$,
i.e.
$$\mathsf{M}(g)\geq \mathsf{N}(g) \geq \mathsf{M}^{\cert}(g)$$
for every Boolean (possibly partial) function $g$. Applying \cref{cert monotonicity} from \Cref{toolbox} on the hypothesis for the first inequality, applying the hypothesis for the second inequality, and applying \cref{cert idempotent} from \Cref{toolbox} for the last equality, we have
    $$\mathsf{M}^\cert(f)
    \geq \mathsf{N}^\cert(f)
    = \max\limits_{z\in\D} \mathsf{N}\of*{f^\cert_z}
    \geq \max\limits_{z\in\D} \mathsf{M}^\cert\of*{f^\cert_z}
    = \of*{\mathsf{M}^\cert}^\cert(f)
    = \mathsf{M}^\cert(f)$$
    so
    $\mathsf{M}^\cert(f) = \mathsf{N}^\cert(f)$.

\end{proof}

We (informally) say a complexity measure $\mathsf{M}$ is {\em local}
if it is defined as the worst case, over inputs $x$, of the complexity at $x$, usually written as $\mathsf{M}(f;x)$.

\begin{remark}[Certification of local complexity measures]\label{local cert = local}
    If a complexity measure $\mathsf{M}$ is \emph{local}, %i.e. 
    and for any Boolean (possibly partial) function
    $g : \D \subseteq \zone^n \to \zone$
    and any
    $z \in \D$
    it verifies
    $\mathsf{M}(g,z) =\mathsf{M}\of*{g^\cert_z}$,
    then for any Boolean (possibly partial) function $f$ we have
    $\mathsf{M}^\cert(f) = \max\limits_{z\in\D} \mathsf{M}\of*{g^\cert_z} = \mathsf{M}(f)$.

    Of course this is exactly \Cref{def:cert qc certification version} (so it is trivial). The remark is actually about noticing that $\s$, $\bs$, $\FC$, $\fbs$, $\RC$ and $\cc$ are usually defined in this local way, e.g.
    $\s(f) := \max\limits_{z\in\D} \s(f,z)$
    where 
    $\s(f,z) = \s\of*{f^\cert_z}$. So for any Boolean (possibly partial) function $f$
    \begin{itemize}
        \item $\s^\cert(f) = \s(f)$
        \item $\bs^\cert(f) = \bs(f)$
        \item $\FC^\cert(f) = \FC(f) = \fbs(f) = \fbs^\cert(f)$, where the second equality follows from~\cites[Theorem~2]{Tal13}[Section~2.6]{GSS16}
        \item $\RC^\cert(f) = \RC(f)$
        \item $\cc^\cert(f) = \cc(f)$.
    \end{itemize}
    
    Also recall that any Boolean (possibly partial) function $f$ verifies
    $\FC(f) = \Theta(\RC(f))$ \cites[Claim~5.5]{Tal13}[Claim~38]{GSS16}
    (even though the claims are not explicitly stated to hold for partial functions, their proofs are done for partial functions).
\end{remark}

\subsection{Certification, witnesses and postselection}

There are a few other candidates to consider if we wish to capture what certificate complexity should mean over various models of computation, or more broadly, for complexity measures $\mathsf{M}$. We have proposed an operational definition, but one might also like to define a certificate ``object'' : given an input $x$, a purported value for $f(x)$, and given this certificate ``object'' for $x$, we would like to use it to verify that $f(x)$ is as claimed. For $\dqc$, this is a partial assignment. For $\rqc$, it can be a weighted set of queries. One certificate may be used to certify several instances of $f$. The associated certificate complexity would then be the computational cost in model $\mathsf M$ to carry out the certification algorithm. One can also think of it as a nondeterministic task: given a certificate and an input, accept if the certificate is valid for this input, otherwise reject. 
Similarly, in a query-based model, one might also phrase it in terms of postselection: guess a certificate, and if the certification procedure accepts, accept; else abort.

For classical query model these three notions coincide. In \Cref{nondeterminism} we make this claim more formal. 

But in the quantum case, we show that they do not coincide. On the postselection side, Cade showed that exact postselected randomized query complexity is equal to  certificate complexity \cite[Theorem 16]{Cade20},
whereas Mahadev and de~Wolf showed that exact postselected quantum query complexity is characterized by rational degree~\cite[Theorem~1,Theorem~2]{MdW15}. On the certification side, we recover that deterministic certification complexity is equal to certificate complexity, whereas our approach allows us to find a distinction between the certification task and postselection in the quantum case. We show that exact quantum query certification complexity $\qec^\cert$ and exact postselected quantum query complexity are two distinct notions, thanks to our lower bound on $\qec^\cert$ by $\deg^\cert$ for which we provide an asymptotic separation with $\rdeg$ in \Cref{thm:separation (deg cert)^2.5 < ndeg}.

\section{Characterizations of certificate complexity}

In this section, we give the proofs of the new characterizations of $\cc$ and $\RC$.

\subsection{Certification for classical complexity measures}
We  start by reproducing the known characterization for the classical certificate complexity. Note that most of the characterizations are already known in the literature ( some of them implicitly, some of them explicitly). Our certification framework is abstract enough to capture the different variant of the same complexity measure. 

\CharacterizationClassicalCertificateComplexities*
\begin{proof}[Proof of \cref{item:characterizations of C} of \Cref{characterizations of C and RC}]
    We  break the proof into three parts.
    \begin{enumerate}[label=(\roman*)]
        \item $\dqc^\cert = \roc^\cert  = \cc:$
        It is known that
        $\dqc \geq \roc \geq \cc$
        (follows from the definitions). Applying \cref{cert monotonicity} from \Cref{toolbox} (certification monotonicity) to the known relation 
        $\dqc \geq \roc \geq \cc$
        one gets
        $\dqc^\cert(f) \geq \roc^\cert \geq \cc^\cert(f) = \cc(f)$
        where the equality follows from \Cref{local cert = local}.
        
        To prove the converse inequality
        $\dqc^\cert(f) \leq \cc(f)$,
        let
        $z\in\D$
        and consider the certificate given by $\cc\of*{f,z}$. The following deterministic algorithm computes $f^\cert_z$ in
        $\cc\of*{f^\cert_z}$ queries: query the defined positions of the certificate, then output $1$ iff the queried values all match the certificate. Then
        $\dqc\of*{f^\cert_z} \leq \cc\of*{f,z}$.
        So
        $$\dqc^\cert(f)
        = \max\limits_{z\in\D} \dqc\of*{f^\cert_z}
        \leq \max\limits_{z\in\D} \cc(f,z)
        = \cc(f).$$
        
        \item 
        $\UC^\cert(f) = \dqc^\cert(f):$
        Apply \cref{A > B > A cert = B cert} from \Cref{toolbox} (Sandwich Lemma) to $\dqc \geq \UC \geq \cc = \dqc^\cert$ where the inequalities follow from the definitions and the equality follows from the previous item.
        
        \item \Cref{C = (C_0 C_1) cert} states that
        $\of*{\cc_0\cdot\cc_1}^\cert = \cc$
        on Boolean (possibly partial) functions, and is proved in \Cref{appendix:proof of lemmas}. So one can conclude by applying \cref{A > B > A cert = B cert} from \Cref{toolbox} (Sandwich Lemma) to
        $\cc \leq \roc \leq O(\CG) \leq O(\cc_0 \cdot \cc_1)$
        where the second inequality follows from \cites[Theorem~1]{JKK+20}[Theorem~22.2]{CGL+23} (the proof of which indeed holds for partial functions, even though not explicitly stated as such in the theorems) and the last inequality follows from \cites[Theorem~22.4]{CGL+23}.
        \qedhere
    \end{enumerate}
\end{proof}

\begin{proof}[Proof of \cref{item:characterizations of RC} of \Cref{characterizations of C and RC}]
    Equalities and inequalities between complexity measures are up to $O(1)$ multiplicative factor. We proceed in four steps:
    \begin{enumerate}[label=(\roman*)]
        \item $\RC(f) = \rqc^\cert(f)$ is shown in \cites[Proposition~4]{Aar08}[Definition~7]{BK19}.
        
        \item $\rqc^\cert(f) = \RS^\cert(f) = \FC^\cert:$
        Apply \cref{A > B > A cert = B cert} from \Cref{toolbox} (Sandwich Lemma) to
        $\rqc \geq \RS \geq \FC = \RC = \rqc^\cert$
        where the first equality follows from the end of \Cref{local cert = local} and the last equality follows from the previous item.
        
        \item $\prt^\cert(f) = \CAdv^\cert(f) = \cfbs^\cert(f) = \FC(f):$
        Apply \cref{A > B > A cert = B cert} from \Cref{toolbox} to
        $\rqc \geq \prt \geq \CAdv \geq \cfbs \geq \fbs = \RC = \rqc^\cert$
        where the first inequality follows from \cite[Theorem~3.1]{JK10}, the second inequality follows from \cites[Theorem~3.5, Definition~19]{JK10}[where $\CAdv$ is denoted $\mathrm{CMM}$, see Lemma~7 for the equivalence with our \Cref{def:CAdv} of $\CAdv$]{AKP+21}, the third inequality follows from \cite[Lemma~26]{ABK21}, the fourth inequality follows from definitions \cite[beginning of Section~2.3.2]{ABK21}, the first equality follows from \Cref{local cert = local} and the last equality follows from previous items.
        
        \item $\EC^\cert(f) = \FC(f):$
        Let $\of*{v_z(i)}_{z\in\D,i\in[n]}$ be optimal weights for the $\FC(f)$ program. Then, for
        $z \in \D$
        and
        $x \in \set{z} \cup \comprehension{y \in \D}{f(y) \neq f(z)}$,
        if
        $f(x) \neq f(z)$
        let
        $B^z_x \subseteq \comprehension{i \in [n]}{x_i \neq z_i}$
        be a minimal sensitive bloc of $z$, and for
        $i \in [n]$
        let
        $$w^z_x(i) := \begin{cases}
            v_x(i)
            &\predicate{if } x=z\\
            \One_{i \in B^z_x}
            &\predicate{if } f(x) \neq f(z)\\
        \end{cases}$$
        Then for all
        $z \in \D$,
        weights $\of*{w^z_x(i)}$ satisfy the $\EC\of*{f^\cert_z}$ program. So
        \begin{align*}
            \EC^\cert(f)
            &= \max\limits_{z\in\D} \EC\of*{f^\cert_z}\\
            &\leq \max\limits_{z\in\D} \max\limits_{x\in\set{z}\cup\comprehension{y\in\D}{f(y) \neq f(z)}} \sum\limits_{i\in[n]} w^z_x(i)\\
            &\leq \max\limits_{z\in\D} \max\set*{\sum\limits_{i\in[n]} v_z(i),\limitB{\max}{x\in\D:}{f(x) \neq f(z)} \card*{B^z_x}}\\
            &\leq \max\limits_{z\in\D} \max\set*{\FC(f),\s(f)}\\
            &\leq \FC(f)\\
            &\leq \EC(f)
        \end{align*}
        where the third inequality follows from the fact that each bit in $B^z_x$ is sensitive for the input $z^{B^z_x}$ as $B^z_x$ is a minimal sensitive bloc of $z$, and the last inequality follows from \cite[Theorem~1]{JKK+20}.
        
        Applying \cref{A > B > A cert = B cert} from \Cref{toolbox} (Sandwich Lemma) and \Cref{local cert = local}, conclude that
        $\EC^\cert = \FC$.\qedhere
    \end{enumerate}
\end{proof}

\subsection{Quantum certification complexity measures}\label{sec:quantum cert measures}
In this subsection we provide the proofs for the  characterizations of  quantum certificate complexity. Before going into the theorem we will first describe some of the definitions of complexity measures that we will be using. 
\begin{defi}[quantum distinguishing complexity~\cite{BK19}]
\label{def:QD}
    Let
    $f : \D \subseteq \zone^n \to \zone$
    be a Boolean (possibly partial) function.
    The \emph{quantum distinguishing complexity of $f$}, denoted $\QD(f)$, is the smallest
    integer $k$ for which there exists a $k$-query quantum algorithm that, on every input
    $x \in D$,
    outputs a quantum state $\rho_x$ such that
    $$\card*{\rho_x-\rho_y}_{\mathrm{tr}} \geq \frac{1}{6}$$
    for every
    $x,y \in \D$
    satisfying
    $f(x) \neq f(y)$.
\end{defi}

\begin{defi}[quantum sabotage complexity~\cite{CMP24}]
    For inputs
    $x,y \in \zone^n$
    and
    $b \in \set*{\ast,\dagger}$,
    define
    $\bracket*{x,y,b} \in \set*{0,1,b}^n$
    coordinate wise by
    $$\bracket*{x,y,b}_i := \begin{cases}
        x_i
        &\predicate{if } x_i=y_i\\
        b
        &\predicate{if } x_i \neq y_i\\
    \end{cases}$$
    Let
    $f : \D \subseteq \zone^n \to \zone$
    be a Boolean (possibly partial) function, and let
    $$\D_0 := f^{-1}(0)
    \qquad\predicate{and}\qquad
    \D_1 := f^{-1}(1).$$
    For
    $b \in \set*{\ast,\dagger}$,
    define the set of \emph{$b$-sabotaged inputs} as
    $$S_b := \comprehension{\bracket*{x,y,b}}{x \in \D_0,\ y \in D_1}$$
    and let
    $$\D_{\mathrm{sab}}:=S_*\cup S_\dagger.$$
    The \emph{sabotage function} associated with $f$ is the partial function
    $$f_\sab : \D_\sab \to \set*{0,1}$$
    defined by
    $$f_\sab(z) := \begin{cases}
        0 & \predicate{if } z\in S_*\\
        1 & \predicate{if } z\in S_\dagger
    \end{cases}$$
    The \emph{quantum sabotage complexity} of $f$ in the weak input model is
    $$\QS_\weak(f) := \qqc\of*{f_\sab}$$
    where $\qqc(\cdot)$ denotes bounded-error quantum query complexity.
\end{defi}

The paper defines several variants
$$
    \QS_\weak(f) := \qqc\of*{f_\sab}
    \qquad\qquad
    \QS_{\mathrm{weak}}^{\mathrm{ind}}(f)
        := \qqc(f_{\mathrm{sab}}^{\mathrm{ind}})
$$
and, in the strong input model,
$$
    \QS_\str(f) := \qqc\of*{f_\sab^\str}
    \qquad\qquad
    \QS_\str^\ind(f) := \qqc\of*{f_\sab^{\str,\ind}}.
$$
See \cite{CMP24} for more details.

\ComputationalCharacterizationQuantumCertificateComplexity*

\begin{proof}
Equalities and inequalities between complexity measures are up to $O(1)$ multiplicative factor. The proof is divided into six parts.
\begin{enumerate}[label=(\roman*), ref=(\roman*)]
    \item\label{proofQC: Q cert = QC} \cites[Proposition~4]{Aar08}[Definition~7]{BK19} proved that any Boolean (possibly partial) function $f$ verifies
    $$\QC(f) = \qqc^\cert(f).$$
    
    \item\label{proofQC: Q cert = QD cert = MM cert} 
    $\qqc^\cert(f) = \QD^\cert(f) = \MM^\cert(f):$
    Apply \cref{A > B > A cert = B cert} from \Cref{toolbox} to
    $\qqc \geq \QD \geq \MM \geq \sqrt{\CAdv} \geq \sqrt{\cfbs} \geq \sqrt{\fbs} = \sqrt{\RC} = \QC = \qqc^\cert$,
    where the first and second inequalities follow from \cite{BK19}, the third inequality follows from \cite[Theorem~30]{ABK21}, the fourth inequality follows from \cite[Lemma~26]{ABK21}, the fifth inequality follows from definitions \cite[beginning of Section~2.3.2]{ABK21}, the first equality follows from \Cref{local cert = local}, the second equality follows from \cite[Theorem~7]{Aar08}, and the last equality follows from the previous \cref{proofQC: Q cert = QC}.
        
    \item\label{proofQC bonus: Q cert = QD cert} As a bonus, we give an elementary proof of
    $$\qqc^\cert(f) = \QD^\cert(f).$$
    For any
    $z \in \D$
    let algorithm $A(z)$ be optimal for $\QD\of*{f^\cert_z}$. Measuring the output of algorithm $A(z)$ on basis $\of*{A(z)\ket{z},\of*{A(z)\ket{z}}^\bot}$ gives an algorithm computing $f^\cert_z$ in the query model of computation $\qqc$ with the same quantum query complexity. So
    $\QD\of*{f^\cert_z} \geq \qqc\of*{f^\cert_z}$.
    Thus
    $$\QD^\cert(f)
    = \max\limits_{z\in\D} \QD\of*{f^\cert_z}
    \geq \max\limits_{z\in\D} \qqc\of*{f^\cert_z}
    = \qqc^\cert(f).$$
    The converse
    $\QD^\cert(f) \leq \qqc^\cert(f)$
    follows from
    $\QD \leq \qqc$
    by \cref{cert monotonicity} from \Cref{toolbox}.
    
    \item\label{proofQC: QSind-weak < QD cert} We give two proofs of the fact that
    $\of*{\QS^\ind_\weak}^{\cert}(f) \leq \QD^\cert(f)$.
    \begin{itemize}
        \item 
        Let $z$ be the input for which the maximum of $\QS^\ind_\weak(f_z)$ is achieved and let $c = f(z)$.
        Let $a = *$ and $b = \dagger$.
        Modify the query oracle so that if the result of the query on the $i$-th bit is $a$ or $b$, then the modified result is $z_i$ and $1-z_i$, respectively.
        In this way, the inputs in $S_{a}$ are all mapped to $z$, and the inputs in $S_{b}$ are all mapped to the inputs $f^{-1}(c)$.
        By \cite{Aar08}, there exists a bounded-error quantum algorithm that uses at most $O(\QC(f))$ queries such that:
        \begin{itemize}
            \item it ``accepts'' $z$;
            \item it ``rejects'' any input $y \in f^{-1}(c)$ and returns a position $i$ such that $z_i \neq y_i$.
        \end{itemize}
        The output of this algorithm can be used to solve the sabotaged function.
        If the algorithm rejected, return $1$; also return the index found by the algorithm.
        If the algorithm accepted, try the same strategy with $a = \dagger$ and $b = *$.
        Since in one of the two cases the sabotaged input will be mapped to an input $y$ such that $f(y) \neq f(z)$, the algorithm will reject in that case and thus solve the sabotaged function.
        The required follows as $\QD^{\cert}(f) = \QC(f)$ by points above.
        
        \item Proof by amplitude amplification: Let
        $z \in \D$.
        Let $A$ be an optimal algorithm for $\QD\of*{f^\cert_z}$, hence running in
        $T = \QD\of*{f^\cert_z}$
        queries.
        
        Consider the following algorithm $B$, which runs on a $b$-sabotaged input $\bracket*{y,z,b}$ (where
        $b \in \set{\ast,\dagger}$
        and
        $f(y) \neq f(z)$)
        of $f^\cert_z$ as follows. Pick
        $t \in [0,T]$
        uniformly at random and run $A^z$ (which does not require any query to $y$) until right before the $t$th query. Measure the query index register -- it gives an index
        $i \in [n]$.
        
        According to \cite[Lemma~12]{BK19}, the probability that
        $y_i \neq z_i$
        (i.e. ${\bracket*{y,z,b}}_i \in \set{\ast,\dagger}$)
        is $\Omega\of*{\frac{1}{T^2}}$, which can be tested in one query to $\bracket*{y,z,b}$. Hence, using quantum amplitude amplification, one can find a $b$-index in $\bracket*{y,z,b}$ with $O(T)$ repetitions of algorithm $B$ and so $O(T)$ queries to $\bracket*{y,z,b}$. Thus,
        $\QS^\ind_\weak\of*{f^\cert_z} \leq \QD\of*{f^\cert_z}$.
        So
        $\of*{\QS^\ind_\weak}^\cert(f) \leq \QD^\cert(f)$.
    \end{itemize}
    
    \item\label{proofQC: QC < QSstr} $\QC \leq \QS_\str(f):$
    \cite[Theorem~5]{CMP24} states that
    $\sqrt{\fbs} \leq \QS_\str$
    and its proof indeed holds for partial functions (see \cite[Lemma~18 and beginning of Section~5.1]{CMP24}), even though it is not explicitly stated as such in the theorem. Conclude with
    $\QC = \sqrt{\FC} = \sqrt{\fbs}$
    where the first equality follows from \cite[Theorem~7]{Aar08} and the last equality from \Cref{local cert = local}.
    
    \item\label{proofQC: (e)} \cref{proofQC: Q cert = QC}, \cref{proofQC: Q cert = QD cert = MM cert}, \cref{proofQC: QSind-weak < QD cert} and \cref{proofQC: QC < QSstr} from \Cref{characterizations of QC: computational} imply that
    $\of*{\QS^\ind_\weak}^\cert \leq \QD^\cert = \qqc^\cert = \QC \leq \QS_\str \leq \QS^\ind_\str \leq \QS^\ind_\weak$
    and
    $\of*{\QS^\ind_\weak}^\cert \leq \QD^\cert = \qqc^\cert = \QC \leq \QS_\str \leq \QS_\weak \leq \QS^\ind_\weak$. Then, applying \cref{A > B > A cert = B cert} from \Cref{toolbox} (Sandwich Lemma) one gets
    $$\of*{\QS_\weak}^\cert = \of*{\QS_\str}^\cert = \of*{\QS^\ind_\weak}^\cert = \of*{\QS^\ind_\str}^\cert = \of*{\qqc^\cert}^\cert = \qqc^\cert$$
    where the last equality follows from \cref{cert idempotent} from \Cref{toolbox}.
    \end{enumerate}
\end{proof}

\subsection{Certification for sensitivity}

Since sensitivity is a local measure, certification applied to sensitivity.
However when applied to spectral sensitivity, we show that $\lambda^\cert$ equals $\sqrt{\s}$:

\begin{theorem}\label{lambda cert = sqrt(s)}
    For any boolean (possibly partial) function
    $f:\D \subseteq \zone^n \to \zone$
    $$\lambda^\cert(f) = \sqrt{\s(f)}.$$
\end{theorem}
\begin{proof}
    Let
    $z \in \D$
    and
    $b:= 1-f(z)$.
    For
    $x,y \in \set{z} \cup f^{-1}(b)$
    let
    $$M_{f^\cert_z}(x,y) := \begin{cases}
        1
        &\predicate{if } f(x) \neq f(y) \predicate{ and there exists } i \in [n] \predicate{ such that } x^i=y\\
        0
        &\predicate{otherwise}
    \end{cases}$$
    Notice that the sensitivity matrix contains only $\s(f;z)$ 1's in the row and column indexed by $z$, and is 0 everywhere else.
    Therefore,
    $\lambda\of*{f^\cert_z} = \norm*{M_{f^\cert_z}} = \sqrt{\s(f,z)}$,
    which implies that 
    $\lambda^\cert(f) = \sqrt{\s(f)}$.
\end{proof}

\section{Certification complexity for $\qoc$ and $\qec$}\label{sec:equality Q_0 cert = Q_E cert}
In this section, we will prove the equivalence between $\qec^\cert$ and $\qoc^\cert$. Let us restate our theorem.

\QZeroCertEqualsQECert*

The proof uses {\em exact} amplitude amplification, {\cite[Theorem~4, second method]{BPM+02}}.

Recall that amplitude amplification is a quantum algorithm that amplifies the correctness of an algorithm that succeeds with probability at least $p$ to a constant success probability, using $O(\sqrt{\frac 1 p})$ iterations~\cite{BPM+02}.

Stated informally, \cite{BPM+02} also show that {\em exact} quantum amplitude amplification can be achieved when the probability of success on an input $z$ is known exactly (not just a lower bound).
In the certification setting, the certification input $z$ is known, and we are provided with a quantum algorithm for $f^\cert_z$ which is zero-error, but not exact, which means it can fail (output $\bot$) with constant probability, but never errs otherwise. The key is to observe that since $z$ is known, and the algorithm is known, the error probability on $z$ is known exactly, so zero error can be incurred on input $z$ by using exact error amplification. In addition, for inputs $y$ such that $f(z) \neq f(y)$, 
the zero-error quantum algorithm always outputs either 0 or $\bot$.    After running exact amplitude amplification parameterized for the error on $z$, the state of the algorithm on any of the remaining inputs $y$ such that $f(z) \neq f(y)$, will remain in the  subspace corresponding to the 0 or $\bot$ outputs, which is orthogonal to the subspace corresponding to the 1 output.

The more formal statement of exact amplitude amplification  follows.

\begin{theorem}[exact quantum amplitude amplification {\cite[Theorem~4, second method]{BPM+02}}]\label{exact quantum amplitude amplification}
    Let
    $f : \zone^n \to \zone$
    be a Boolean total function. Let
    $\H = \H_0 \otimes \H_1$
    be an Hilbert space of dimension at least $2^n$ and $\pi_0$ (resp. $\pi_1$) be the projector onto $\H_0$ (resp. $\H_1$). Let $\of*{\ket{x}}_{x\in\zone^n}$ be a family of orthonormal states in $\H$ (that are not necessarily in $\H_0$ or $\H_1$). Let $\ket{\text{init}}$ be a state in $\H$. Let $G$ be a unitary. Let the success probability of $G$ be
    $p := \norm*{\pi_1 G\ket{\text{init}}}_2^2$.
    Let $S_0$ be a unitary such that for all
    $\ket{\chi} \in \H_0$
    $$S_0(\phi) \ket{\chi} = e^{i\phi} \ket{\chi}$$
    and for all
    $\ket{\chi} \in \H_1$
    $$S_0(\phi) \ket{\chi} = \ket{\chi}.$$
    Let $S_f(\varphi)$ be a unitary such that for all
    $x \in \zone^n$
    $$S_f(\varphi) \ket{x} = e^{i\varphi f(x)} \ket{x}.$$
    We will refer to $S_0(\phi)$ and $S_f(\varphi)$ as \emph{pseudo-reflections}.
    Let
    $Q(\phi,\varphi) := G S_0(\phi) G^\dagger S_f(\varphi)$.

    If
    $p \neq 0$,
    then there exists phases $\phi$ and $\varphi$ and a number of iterations $k$ such that
    $$Q(\phi,\varphi) Q(\pi,\pi)^k G \ket{\text{init}} \in \H_1$$
    and
    $$k \leq O\of*{\frac{1}{\sqrt{p}}}$$
    i.e. it is an algorithm that achieves exact success, with $O\of*{\frac{1}{\sqrt{p}}}$ queries to $G$ and $G^\dagger$, and relies on the knowledge of $p$.
\end{theorem}

\begin{lemma}[exact quantum certification algorithm from zero-error quantum certification algorithm]
    Let
    $f : \D \subseteq \zone^n \to \zone$
    be a Boolean (possibly partial) function. Let
    $\H = \hat{\H}_0 \otimes \hat{\H}_1 \otimes \hat{\H}_\bot$
    be an Hilbert space of dimension $2^n$ and $\hat{\pi}_0$ (resp. $\hat{\pi}_1$ and $\hat{\pi}_{0,\bot}$) be the projector onto $\hat{\H}_0$ (resp. $\hat{\H}_1$ and $\hat{\H}_0 \oplus \hat{\H}_\bot$). Let $\of*{\ket{x}}_{x\in\zone^n}$ be a family of orthonormal states in $\H$ (that are not necessarily in $\hat{\H}_0$, $\hat{\H}_1$ or $\hat{\H}_\bot$). Let $\ket{\text{init}}$ be a state in $\H$.
    
    Let
    $0 < \varepsilon < 1$.
    Let
    $z \in \D$.
    Let $A$ be a $\qoc$ algorithm computing $f^\cert_z$ in $\qoc\of*{f^\cert_z}$ queries, i.e. for any input
    $x \in \set{z}\cup\comprehension{y \in \D}{f(y) \neq f(z)}$
    the algorithm $A$, when given query access to input $x$ (which we denote $A^x$), has measurement probabilities satisfying
    $\norm*{\hat{\pi}_{1-f(x)} A^x \ket{\text{init}}}_2^2 = 1$
    and
    $\norm*{\hat{\pi}_{f(x)} A^x \ket{\text{init}}}_2^2 \leq \frac{1}{2}$.
    
    Let $S_{0,\bot}$ be a unitary such that for all
    $\ket{\chi} \in \H_0 \otimes \H_\bot$
    $$S_{0,\bot}(\phi) \ket{\chi} = e^{i\phi} \ket{\chi}$$
    and for all
    $\ket{\chi} \in \H_1$
    $$S_{0,\bot}(\phi) \ket{\chi} = \ket{\chi}.$$
    Let $S_f(\varphi)$ be a unitary such that for all
    $x \in \D$
    $$S_f(\varphi) \ket{x} = e^{i\varphi f(x)} \ket{x}.$$
    and acts as the identity on any other $\ket{x}$. For phases $\phi$ and $\varphi$, define algorithm $Q(\phi,\varphi)$, when given query access to input
    $x \in \set{z}\cup\comprehension{y \in \D}{f(y) \neq f(z)}$,
    as
    $$Q^x(\phi,\varphi) := A^x S_{0,\bot}(\phi) \of*{A^x}^\dagger S_f(\varphi)$$
    where in order to compute $\of*{A^x}^\dagger$, the inverse of algorithm $A$ is implemented with the same number of queries as $A$, which is always possible since the quantum query operator is its own inverse.

    Then, the theorem states that there exists phases $\phi$ and $\varphi$ and a number of iterations $k$ such that
    $$Q^z(\phi,\varphi) Q^z(\pi,\pi)^k A^z \ket{\text{init}} \in \hat{\H}_1$$
    and for all
    $y \in D$
    such that
    $f(y) \neq f(z)$
    $$Q^y(\phi,\varphi) Q^y(\pi,\pi)^k A^y \ket{\text{init}} \in \hat{\H}_0 \otimes \hat{\H}_\bot$$
    and
    $$k \leq O\of*{1}$$
    i.e. it is an algorithm that achieves exactness w.r.t. measurement $(\hat{\pi}_1,\hat{\pi}_{0,\bot})$ in $O\of*{\qoc\of*{f^\cert_z}}$ queries.
\end{lemma}
\begin{proof}
    Let us first consider the case of input $z$, which will fix phases $\phi$ and $\varphi$ in order to have exactness, and then prove exactness for any other input under those parameters:
    \begin{itemize}
        \item For input $z$: Apply \Cref{exact quantum amplitude amplification} with
        $\H_1 := \hat{\H}_1$,
        $\H_0 := \hat{\H}_0 \otimes \hat{\H}_\bot$,
        $G := A^z$
        and
        $p := \norm*{\hat{\pi}_1 A^z \ket{\text{init}}}_2^2 \geq \frac{1}{2} > 0$.
        This gives exactness on input $z$ and fixes phases $\phi$ and $\varphi$ and the number of iterations
        $k \leq O\of*{\frac{1}{\sqrt{p}}} \leq O(1)$,
        which corresponds to $O\of*{\qoc\of*{f^\cert_z}}$ queries (since $G$ and $G^\dagger$ are computed using the same number of queries).

        \item For input
        $y \in \D$
        such that
        $f(y) \neq f(z)$:
        Notice that for any phases $\phi'$ and $\varphi'$, operator $Q^y(\phi',\varphi')$ only adds a global phase, i.e. there exists a phase $\gamma$ such that
        $Q(\phi',\varphi') A^y \ket{\text{init}} = e^{i\gamma} A^y \ket{\text{init}}$.
        In particular with phases $\phi$, $\varphi$ and $\pi$ it implies that
        $Q^y(\phi,\varphi) Q^y(\pi,\pi)^k G \ket{\text{init}} \in \hat{\H}_0 \otimes \hat{\H}_\bot$.
        Hence exactness on input $y$.\qedhere
    \end{itemize}
\end{proof}
\noindent\Cref{Q_0 cert = Q_E cert} follows as a corollary.

\section{Polynomial certification complexity }

In this section, we give the proofs of our results pertaining to the certification of polynomial degree comlexity measures.

\subsection{Exact polynomials}

In order to give a flavor of certifying polynomials, we start by giving  a direct proof that $\deg^\cert(f) \leq \cc(f)$.
(Notice that since $\deg(f) \leq \dqc(f)$, we can also immediately get $ \deg^\cert(f) \leq \cc^\cert(f) = \cc(f)$ by applying~\Cref{toolbox} \cref{cert monotonicity}.)

\begin{remark}[elementary proof of $\deg^\cert(f) \leq \cc(f)$]\label{deg cert < C}
    Let
    $f:\D \subseteq \zone^n \to \zone$
    be a Boolean (possibly partial) function. Then
    $$\deg^\cert(f)
    = \max\limits_{z\in\D} \deg\of*{f^\cert_z}
    \leq \max\limits_{z\in\D} \cc(f,z)
    = \cc(f)$$
    where the inequality follows from the fact that
    $P(x) := \prod\limits_{i \in c} x_i - z_i$
    is an exact polynomial for $f^\cert_z(x)$ whenever $c$ is an $f$-certificate of $z$. \qedhere
\end{remark}

By construction, the degree certification measure $\deg^\cert$ is a lower bound on $\qec^\cert$ -- as one can prove by applying \cref{cert monotonicity} from \Cref{toolbox} to
$\deg \leq 2\qec$~\cite[Theorem~4.3]{BBCM+01}.
Hence, it is a lower bound on $\qoc$, by the collapse $\qec^\cert = \qoc^\cert$ (\Cref{Q_0 cert = Q_E cert}). In \Cref{sec:osdeg} we will show another way to prove that degree certification gives a lower bound on $\qoc$ by proving that it collapses with its one-sided-error counterpart.

Finally, for completeness, we restate here the upper bounds in terms of $\deg^\cert$ proved in the introduction (\Cref{subsec:polynomial relations}).

\UpperBoundOnCByDegCertBs*

\subsection{Approximate polynomials with two-sided error}

We prove a new characterization of $\QC$ in terms of certification complexity  for approximate polynomials.

\PolynomialCharacterizationQuantumCertificateComplexity*
\begin{proof}
    Equalities and inequalities between complexity measures are up to $O(1)$ multiplicative factor.
    Apply \cref{A > B > A cert = B cert} from \Cref{toolbox} to
    $\qqc^\cert = \QC = \sqrt{\RC} = \sqrt{\fbs} \leq \adeg \leq \bmadeg \leq \qqc$
    where the first equality follows from \Cref{characterizations of QC: computational}, the second equality follows from \cite[Theorem~7]{Aar08}, the third equality follows from \Cref{local cert = local}, the first inequality follows from \cite[Theorem~28]{ABK21}, the second inequality follows from definitions, and the last inequality follows from \cites[Lemma~10]{AA15}[Lemma~2.2]{AAI+16}.
\end{proof}

\subsection{Approximate polynomials with one-sided error}\label{sec:osdeg}

We introduce one-sided-error approximate degree as the zero-error counterpart of degree and approximate degree, hence a new lower bound on zero-error quantum query complexity (see \Cref{def:deg measures}).

\begin{fact}\label{adeg < osdeg < deg}
    For any Boolean (possibly partial) function $f$
    $$\adeg(f) \leq \osdeg(f) \leq \deg(f).$$
\end{fact}
\begin{proof}
    This follows directly from their definitions.
\end{proof}

\LowerBoundOnQZeroByOsdeg*
\begin{proof}
    Let $A$ be a $T$-query zero-error quantum query algorithm for $f$. Then according to \cite[Lemma~4.2]{BBCM+01} there exist multilinear real polynomials $P^0$ and $P^1$ of degree at most $2T$ such that for all
    $x \in \D$
    $$P^0(x) = \Pr\bracket*{A\predicate{ outputs }0\predicate{ on }x}
    \qquad\predicate{and}\qquad
    P^1(x) = \Pr\bracket*{A\predicate{ outputs }1\predicate{ on }x}.$$
    Then, not only each $P^b$ is a $\frac{1}{2}$-approximating bounded polynomial for $f$, but also, since $A$ is zero-error, it is moreover $b$-sided-error. So
    $$\osdeg[^0_{\frac{1}{2}}](f) \leq \deg\of*{P^0} \leq 2T
    \qquad\predicate{and}\qquad
    \osdeg[^1_{\frac{1}{2}}](f) \leq \deg\of*{P^1} \leq 2T.$$
    Thus
    $\osdeg[_{\frac{1}{2}}](f) \leq 2T \leq 2\qoc(f)$
    and more generally for any
    $0 < \varepsilon \leq \frac{1}{2}$. Similarly, using standard error reduction of a $\qoc$ algorithm by iterating it a constant number of time, one gets
    $\qoc(f) \geq \Omega\of*{\osdeg[_\varepsilon](f)}$.
\end{proof}

We can easily derive the following relation on the corresponding certification measures.
\begin{corollary}\label{osdeg cert < Q_0 cert}
    For any
    $0 < \varepsilon < 1$
    and any Boolean (possibly partial) function $f$,
    $$\qoc^\cert(f) \geq \Omega\of*{\osdeg[_\varepsilon]^\cert(f)}.$$
\end{corollary}
\begin{proof}
    Apply \cref{cert monotonicity} from \Cref{toolbox} to \Cref{osdeg < Q_0}.
\end{proof}

Before we prove that the one-sided-error approximate degree and exact degree certification measures collapse, we need the following lemma.

\begin{lemma}\label{deg cert < osdeg}
    For any
    $0 < \varepsilon < 1$
    and any Boolean (possibly partial) function
    $f:\D \subseteq \zone^n \to \zone$
    $$\deg^\cert(f) \leq O\of*{\osdeg[_\varepsilon](f)}.$$
\end{lemma}
\begin{proof}
    Recall
    $\osdeg[_\varepsilon](f) = \max\of*{\osdeg[^0_\varepsilon](f),\osdeg[^1_\varepsilon](f)}$.
    Let $P^0$ and $P^1$ be the corresponding bounded polynomials that approximate $f^\cert_z$ with one-sided error $\varepsilon$, i.e.
    for all
    $x \in f^{-1}(0)$
    they satisfy 
    $P^0(x)=0$
    and
    $\card*{P^1(x)} \leq \varepsilon$,
    and for all
    $y \in f^{-1}(1)$
    they satisfy
    $\card*{P^0(y)-1} \leq \varepsilon$,
    and
    $P^1(y)=1$.
    
    Let
    $z \in \D$
    and
    $b := 1-f(z)$.
    Define polynomial $P_z$ as
    $$P_z(x) := \frac{P^b(x) - b}{P^b(z) - b}.$$
    Then $P_z$ is equal to $f^\cert_z$ on its domain, and is bounded by $\frac{1}{1-\varepsilon}$ since $P^b$ is bounded by $1$ and the denominator satisfies
    $1-\varepsilon \leq \card*{P^b(z)-b}$.
    Thus,
    \begin{align*}
        \osdeg[_\varepsilon]^\cert(f)
        &= \max\limits_{z\in\D} \osdeg[_\varepsilon]\of*{f^\cert_z}\\
        &= \max\limits_{z\in\D} \deg\of*{P_z}\\
        &\geq \max\limits_{z\in\D} \Delta_{0,0,\frac{1}{1-\varepsilon}} \of*{f^\cert_z}\\
        &\geq \Omega\of*{\max\limits_{z\in\D} \deg\of*{f^\cert_z}}\\
        &\geq\Omega\of*{\deg^\cert(f)}\\
    \end{align*}
    where the second inequality follows from \Cref{range-reduction: exact}.
\end{proof}

We now state and prove the collapse of exact and one-sided-error approximate degree certification complexities. This result can be viewed as the degree analogue of the collapse of $\dqc^\cert$ and $\roc^\cert$, as well as 
$\qec^\cert$ with $\qoc^\cert$.

\CollapseDegCertOsdegCert*
\begin{proof}
    Apply \cref{A > B > A cert = B cert} from \Cref{toolbox} to
    $\deg \geq \osdeg[_\varepsilon] \geq \Omega\of*{\deg^\cert}$
    where the last inequality follows from \Cref{deg cert < osdeg}. The lower bound on $\qoc^\cert$ is simply \Cref{osdeg cert < Q_0 cert}.
\end{proof}

Thus, both the exact/zero-error quantum query complexity collapses and the exact/one-sided-error degree collapse are individually sufficient to prove that the degree certification measure is a novel lower bound on $\qoc$.

\begin{note}
    Note that $\OR$ provides a separation showing that $\deg^\cert \neq \adeg^\cert$.
\end{note}

Our final result in this section is that rational degree is a lower bound on one-sided-error approximate degree certification complexity. First, we prove the following lemma.

\begin{lemma}\label{ndeg cert = ndeg}
    For any Boolean (possibly partial) function
    $f:\D \subseteq \zone^n \to \zone$,
    $$\rdeg^\cert(f) = \rdeg(f).$$
\end{lemma}
\begin{proof}
    We use the characterization
    $\rdeg(f) = \max\set*{\ndeg(f),\ndeg(\neg f)}$~\cite[Fact~3]{KKW+26}.
    For
    $z \in \D$
    let $P_z$ be a polynomial such that
    $\deg\of*{P_z} = \ndeg\of*{f^\cert_z}$
    and
    $P_z(z) \neq 0$
    and for any
    $y \in \D$
    such that
    $f(y) \neq f(z)$
    it satisfies
    $P_z(y)=0$.

    Let
    $b \in \zone$
    and
    $z \in f^{-1}(b)$.
    Enumerate 
    $f^{-1}(b) =: \set*{z_1,\ldots,z_k}$
    where
    $z_1 = z$. Let us prove by recurrence on
    $j \in [k]$
    the property $\of*{R_j}$ defined as:
    There exists real numbers $\of*{\alpha_i}_{i\in[j]}$ such that polynomial
    $P_j := \sum\limits_{i\in[j]} \alpha_i \cdot P_{z_i}$
    is non zero on inputs $\set*{z_1,\ldots,z_j}$. For the base case $\of*{R_1}$, let
    $\alpha_1 := 1$.

    Suppose $\of*{R_j}$, which gives real numbers $\of*{\alpha_i}_{i\in[j]}$. Let $A$ be a set of $k+1$ distinct non-zero real numbers. Since
    $P_{z_{j+1}}\of*{z_{j+1}} = 1 \neq 0$
    and
    $P_j$
    is non zero on inputs $\set*{z_1,\ldots,z_j}$, then there exists some
    $\alpha_{j+1} \in A$
    such that
    $P_j + \alpha_{j+1} \cdot P_{z_{j+1}} = P_{j+1}$
    is non-zero on $\set*{z_1,\ldots,z_{j+1}}$, hence proving $\of*{R_{j+1}}$.

    Thus $\of*{R_k}$ is proved, giving
    $P^b:=P_k$
    which is a $b$-nondeterministic polynomial for $f$, i.e. for all
    $x \in f^{-1}(b)$
    it satisfies
    $P^b(x) \neq 0$
    and for all
    $x \in f^{-1}(1-b)$
    it satisfies
    $P^b(x) = 0$.
    
    Thus, having done this for each value of $b$, one gets
    \begin{align*}
        \rdeg(f)
        &\leq \max\limits_{b\in\zone} \deg\of*{P^b}\\
        &\leq \max\limits_{b\in\zone} \max\limits_{z\in f^{-1}(b)} \deg\of*{P_z}\\
        &\leq \max\limits_{z\in\D} \rdeg\of*{f^\cert_z}\\
        &= \rdeg^\cert(f)\\
        &= \max\limits_{z\in\D} \rdeg\of*{f^\cert_z}\\
        &= \max \set*{\max\limits_{z\in f^{-1}(0)} \max\set*{\ndeg\of*{f^\cert_z},\ndeg\of*{\neg f^\cert_z}}, \max\limits_{z\in f^{-1}(1)} \max\set*{\ndeg\of*{f^\cert_z},\ndeg\of*{\neg f^\cert_z}}}\\
        &\leq \max \set*{\max\limits_{z\in f^{-1}(0)} \max\set*{\ndeg(\neg f),\ndeg(f)}, \max\limits_{z\in f^{-1}(1)} \max\set*{\ndeg(f),\ndeg(\neg f)}}\\
        &= \max\set*{\ndeg(f),\ndeg(\neg f)}\\
        &= \rdeg(f)
    \end{align*}
    Hence
    $\rdeg(f) = \rdeg^\cert(f)$.
\end{proof}

\begin{corollary}\label{ndeg < osdeg cert}
    For any
    $0 < \varepsilon < 1$
    and any Boolean (possibly partial) function $f$
    $$\rdeg(f) \leq \osdeg[_\varepsilon]^\cert(f).$$
\end{corollary}
\begin{proof}
    One gets
    $\rdeg^\cert \leq \osdeg[_\varepsilon]^\cert$
    by applying \cref{cert monotonicity} from \Cref{toolbox} to
    $$\rdeg(f) = \max\set*{\ndeg(f),\ndeg(\neg f)}
    \leq \max\set*{\osdeg[^1_\varepsilon](f),\osdeg[^0_\varepsilon](f)}
    = \osdeg[_\varepsilon](f)$$
    where the first equality follows from \cite[Fact~3]{KKW+26}. Then conclude with \Cref{ndeg cert = ndeg}.
\end{proof}
\subsection{Separations}
\subsubsection{Separations between degree certification and degree measures}

The $\Tribes$ function  is a natural candidate to separate measures from their certification measure, since $\Tribes^\cert_x$ reduces to solving an $\OR$ function on a block (for $0$-instances) or on one variable per block (for $1$-instances). This holds for any measure that is stable. For degree and approximate degree we get the following separations.

\begin{proposition}[quadratic separation between $\deg$ and $\deg^\cert$ and between $\adeg$ and $\adeg^\cert$ by a total function]\label{separation: deg^2 < deg cert and adeg^2 < adeg cert}
    ~
    \begin{enumerate}
        \item   $\deg\of*{\Tribes_{\sqrt{n},\sqrt{n}}} = n$
        and
        $\deg^\cert\of*{\Tribes_{\sqrt{n},\sqrt{n}}} = \deg\of*{\OR_{\sqrt{n}}} = \sqrt{n}$.
        \item    $\adeg\of*{\Tribes_{\sqrt{n},\sqrt{n}}} = \sqrt{n}$
        and
        $\adeg^\cert\of*{\Tribes_{\sqrt{n},\sqrt{n}}} = \adeg\of*{\OR_{\sqrt{n}}} = n^\frac{1}{4}$.
    \end{enumerate}
\end{proposition}

\subsubsection{Separation between $\deg^\cert$ and $\rdeg$}\label{sec: ndeg degcert separation}
 {\em The proofs in this section were obtained using significant assistance from GPT‑5.6 Sol.}

Finally, we show that exact degree certification complexity is asymptotically different from rational degree.
This is exhibited by the following family of functions.
Define
$$f_m := \AND_m \circ \of*{\lnot \Exact_{1,m}}$$
with variables $x_{i,1},\ldots,x_{i,m}$ in the $i$-th block.
Here,
$\Exact_{1,m}\of*{z_1,\ldots,z_m} = 1$
iff
$\card*{z}=1$,
where $\card*{z}$ is the Hamming weight of $z$.

\SeparationRdegDegCert*

\begin{proof}
    First we deal with the rational degree and then proceed with the degree certification lower bound.

    \paragraph{Rational degree.}
    We prove that
    $\rdeg\of*{f_m} = m$
    via the non-deterministic degree characterization
    $\rdeg\of*{f_m} = \max\set*{\ndeg\of*{f_m},\ndeg\of*{\neg f_m}}$~\cite[Fact~3]{KKW+26}.
    For $\ndeg\of*{f_m}$, we have a non-deterministic polynomial $\prod\limits_{i=1}^m \of*{1-\card*{x_i}}$.
    For $\ndeg\of*{\neg f_m}$, first define
    $Q(w) := w\prod\limits_{i=2}^m (w-i)$.
    Then the polynomial
    $\sum\limits_{i=1}^m Q\of*{\card*{x_i}}$
    is a non-deterministic polynomial for $\neg f_m$.
    The degree of each polynomial is at most $m$, so
    $\rdeg(f) \leq m$.

    For the other direction, fix all the variables except one in each block to $0$.
    This restricted function is then $\neg \OR_m$.
    Since
    $\rdeg(f) = \max\set*{\ndeg(f),\ndeg(\neg f)}$~\cite[Fact~3]{KKW+26},
    and the non-deterministic degree is non-increasing under restrictions, and $\rdeg\of*{\OR_m} = m$ (see e.g.~\cite[Section~2.1]{dW03}), then we have $\rdeg(f) \geq m$.
    
    \paragraph{Degree lower bound.}
    For the $\deg^{\cert}$ lower bound, we use the notion of selective families \cite{CMS01}.
    An $(n,k)$-selective family of sets is a collection of subsets $S_1, \ldots, S_\ell \subseteq [n]$ such that for any subset $T \subseteq [n]$ with $1 \leq |T| \leq k$, we have $|S_i \cap T| = 1$ for at least one $i \in [\ell]$.
    A $s$-light selective family is such that $|S_i| \leq s$ for all $i \in [\ell]$.
    \begin{lemma} \label{lem:selective}
        There exists an $s$-light $(n,k)$-selective family of sets with size $\ell \leq O\of*{\of*{k+\frac{n}{s}}\cdot{\log(n)}^2}$.
    \end{lemma}

    \begin{proof}
        In \cite{Hra21}, an $s$-light $(n,\omega)$-selector is a collection of subsets $S_1, \ldots, S_\ell \subseteq [n]$ such that for any subset $T \subseteq [n]$ with $\frac{\omega}{2} \leq \card*{T} \leq \omega$, we have $\card*{S_i \cap T} = 1$ for at least $\frac{\omega}{4}$ elements $i \in [\ell]$, and $\card*{S_i} \leq s$ for all $i \in \ell$.
        This is stronger than selective families, but is sufficient for our needs.
        Note the difference here that $\card*{T} \geq \frac{\omega}{2}$.
        Theorem 4 of the same work shows the existence of an $s$-light $(n,\omega)$-selector of size $O((\omega+n/s)\log n)$.
        By taking the union of such selectors for $\omega = k, \frac{k}{2}, \frac{k}{4}, \ldots, 2$, we obtain an $s$-light $(n,k)$-selective family of size $\ell \leq O\of*{\of*{k+\frac{n}{s}}{\log(n)}^2} \leq {\widetilde O}\of*{k+\frac{n}{s}}$.
    \end{proof}

    Now take $s = m$, $n \leq O\of*{\frac{m^2}{{\log(m)}^2}}$, $k \leq O\of*{\frac{m}{\log(m)^2}}$.
    Then by \Cref{lem:selective}, there exists an $m$-light $\of*{O\of*{\frac{m^2}{{\log(m)}^2}}, O\of*{\frac{m}{{\log(m)}^2}}}$-selective family of size
    $\ell \leq O\of*{m+\frac{m^2}{m}} \leq O(m)$.
    By picking appropriate constants, we can assume that
    $\ell \leq m$.

    Now define a partial function $g : \zone^n \to \zone$ on inputs $z_1, \ldots, z_n$, which is $1$ on the input $z = 0^n$ and $0$ on the inputs with $|z| \leq k$.
    We have that $\deg(g) \geq \sqrt{nk}$ by \Cref{lem:poly} (below).
    Now we show that this function can be embedded in $f_m$.

    For each $i \in \bracket*{\ell}$, map the elements of $S_i$ to $|S_i|$ variables in the $i$-th block $x_i$.
    Then for an input $z$, let $x(z)$ be defined as follows.
    Set all of the variables of $x(z)$ to which nothing has been mapped to $0$.
    If $z_j = b$, then all of the variables of $x$ to which $z_j$ has been mapped are equal to $b$ as well.
    By the selective family properties, if $|z|>0$, there is at least one block in which there is a single $1$.
    Hence, $f_m(x) = 0$.
    Finally, if $z = 0$, then $f_m(x) = 1$ because all variables are $0$.

    Now consider the polynomial representing the  $0^{m^2}$-certification function $f_{m,0^{m^2}}^\cert$.
    By replacing the mapped variables of $x$ with the corresponding variables of $z$, and other variables to $0$, we obtain a polynomial for $g$.
    Hence,
    $\deg(g) \leq \deg\of*{f_{m,0^{m^2}}^\cert} \leq \deg^\cert\of*{f_m}$,
    and therefore
    $$
        \deg^\cert\of*{f_m} \geq \Omega\of*{\sqrt{\frac{m^3}{{\log(m)}^4}}} \geq {\widetilde\Omega}\of*{m^{3/2}}.\qedhere
    $$
\end{proof}

It remains to prove \Cref{lem:poly}.
\begin{lemma} \label{lem:poly}
    Let  \(P:\mathbb{R}^n \to \mathbb{R}\) be a polynomial satisfying 
\begin{align*}
P\of*{0^n}=1&\\
P(z)=0 &\qquad\predicate{if }1 \leq |z| \leq k\\
0 \leq P(z) \leq 1
&\qquad\predicate{for all inputs }
z\in\zone^n.
\end{align*}
Then
\[
\deg(P) = \Omega(\sqrt{nk}).
\]
\end{lemma}

\begin{proof}
Symmetrizing $P$, we obtain a univariate polynomial $Q$ with $
\deg(Q) \leq \deg(P) $
such that
\begin{align*}
Q(0)=1& \\
Q(1)=\cdots=Q(k)=0& \\
0\leq Q(j)\leq 1&
\qquad\predicate{for }
j=0,\ldots,n.
\end{align*}
Let $ d=\deg(Q).$ Since $Q$ has $k$ distinct roots, we have $d\geq k$.

Then we can factor this polynomial:
\[
Q(x)
=
\prod_{i=1}^k
\left(1-\frac{x}{i}\right) R(x).
\]
Then $R(0)=1$ and $\deg(R)=d-k$ .

Let $A$ be a large integer constant to be defined later, and set $t=Ak$.
If $t\geq n/2$, then $k=\Omega(n)$, and therefore $d\geq k=\Omega(\sqrt{nk})$.
Thus, we may assume that $t<n/2.$

For every integer $j\geq t$,
\[
\left|
\prod_{i=1}^k
\left(1-\frac{j}{i}\right)
\right|
=
\prod_{i=1}^k \frac{j-i}{i}.
\]
Since $j\geq Ak$, for every $1\leq i\leq k$ we have
\[
\frac{j-i}{i}
\geq
\frac{Ak-k}{k}
=
A-1.
\]
Therefore
\[
\left|
\prod_{i=1}^k
\left(1-\frac{j}{i}\right)
\right|
\geq
(A-1)^k.
\]
Since \(0\leq Q(j)\leq 1\), it follows that for\(j=t,\ldots,n\) we have
$|R(j)|\leq(A-1)^{-k}$.

Define $S(x)=1-R(x)^2$ and $\varepsilon=(A-1)^{-2k}$. Then $S(0)=0$
and for \(j=t,\ldots,n\) we have
\[
1-\varepsilon
\leq
S(j)
\leq
1.
\]
Now let $D:=\deg(S)=2(d-k)$.

We now apply Theorem~17 of \cite{BCdWZ99}: there exist
absolute constants \(a,b>0\) such that, whenever \(D\le n-t\),
\[
\varepsilon
\ge
\frac{1}{a}
\exp\left(
-\frac{bD^2}{n-t}
-\frac{4D\sqrt{tn}}{n-t}
\right).
\]

If \(D>n-t\), then, since \(t<n/2\),
\[
d
=
k+\frac{D}{2}
\ge
\frac{D}{2}
>
\frac{n-t}{2}
>
\frac{n}{4}.
\]
Since \(k\le n\), this implies $d=\Omega(n)=\Omega(\sqrt{nk})$, and we are done.

It remains to consider the case \(D\le n-t\). Substituting 
$t=Ak$
and
$\varepsilon=(A-1)^{-2k}$
into the above inequality and taking logarithms gives
\[
2k\log(A-1)
\leq
\log a
+
\frac{bD^2}{n-t}
+
\frac{4D\sqrt{Akn}}{n-t}.
\]
Since $t<n/2$, we have $n-t\geq n/2$,
and therefore
\[
2k\log(A-1)
\leq
\log a
+
\frac{2bD^2}{n}
+
8\sqrt{A}\,
D\sqrt{\frac{k}{n}}.
\]
Choosing \(A\) to be a sufficiently large, we obtain
for some absolute constant \(C>0\),
\[
k
\leq
C \cdot \of*{
\frac{D^2}{n}
+
D\sqrt{\frac{k}{n}}
}.
\]

Now let
$x:=\frac{D}{\sqrt{nk}}$.
Dividing the preceding inequality by $k$ and replacing
$D=x\sqrt{nk}$
gives
\[
1\leq C \cdot \of*{x^2+x}.
\]
Hence $x=\Omega(1)$, and therefore $D=\Omega(\sqrt{nk})$.
Since $ D=2(d-k)\leq 2d$,  
we conclude that $d=\deg(Q)=\Omega(\sqrt{nk}).$
Finally, since $\deg(Q) \leq \deg(P)$, we have
\[
\deg(P)=\Omega\of*{\sqrt{nk}}.\qedhere
\]
\end{proof}

{
\subsection{Tight lower bound on $\qoc$}

Given that $\deg^\cert$ is a lower bound on $\qoc$ (\Cref{osdeg < Q_0}, \Cref{deg cert = osdeg cert} and \cref{cert is easier} from \Cref{toolbox}), as an immediate corollary of \Cref{thm:separation (deg cert)^2.5 < ndeg} we obtain that
$$\qoc(f_m) = \widetilde\Omega\of*{n^{{3}/{2}}}.$$
This lower bound is tight for this function.

\begin{lemma}\label{thm:Q_0 algo}
     ~
    \begin{enumerate}
    \item $\qoc(f_m) = O\of*{m^{{3}/{2}}}$.
    \item $\qqc(f_m) = O(m)$.
    \end{enumerate}
\end{lemma} 

\begin{proof}
    For each $b \in \zone$, we show an algorithm that outputs $b$ on $b$-inputs with constant probability, never outputs $b$ on any $(1{-}b)$-input, and uses $O\of*{m^{{3}/{2}}}$ queries.
    By running the two algorithms in alternation repeatedly, we obtain a zero-error quantum algorithm for $f_m$.

    \paragraph{Positive inputs.}
    First we show the case for the $1$-inputs.
    For an input $x$ to be positive, each of the $\Exact_{1,m}$ functions have to evaluate to $0$.
    We build a one-sided procedure with $O\of*{m^{3/2}}$ expected number of queries for the inputs of $\Exact_{1,m}$.
    By Markov's inequality, we can stop the algorithm after $O\of*{m^{3/2}}$ queries to get a bounded-error one-sided algorithm.

    Let $z_1, \ldots, z_m$ be the input to an instance of $\Exact_{1,m}$.
    For $z$ to be a $0$-input, we need to have either $|z| = 0$ or $|z| \geq 2$.
    First we run the exact Grover's search with a promise that there is a unique marked element \cite{Hoy00}.
    If that returns $0$, then we know with certainty that $|z| \neq 1$, so in that case we return the correct answer.
    If we get $1$, then $|z| \geq 1$.
    To determine whether $|z| \geq 2$, we search for two distinct positions $i, j$ such that
    $z_i = z_j = 1$
    using two applications of Grover's search.
    If this doesn't find the two ones, we repeat.
    If we find the two ones, we query the two positions to certify that
    $|z| \geq 2$
    with certainty.
    
    If
    $|z| \geq 2$,
    this procedure evaluates $\Exact_{1,m}(z)$ to $0$ in $O\of*{\sqrt{m}}$ expected number of queries.
    If $|z| = 1$, we never output $0$ as the certification of the two $1$s will be unsuccessful.
    If $|z| = 0$, we evaluate $\Exact_{1,m}(z)$ correctly already after the exact Grover's search.
    By running this procedure for all $m$ input blocks of $f_m$, we get the required algorithm.
    
    \paragraph{Negative inputs.}
    For an input $x$ to be a $0$-input, there must be at least one $\Exact_{1,m}$ that evaluates to $1$.
    To determine that for a single copy of $\Exact_{1,m}$, with input $z_1, \ldots, z_m$, we can first run Grover's search to find a single $1$, and then run another Grover's search to find whether there are any more $1$s.
    If $|z| = 1$, this gives the correct answer to this $\Exact_{1,m}$ with constant probability.

    To determine whether any copy of $\Exact_{1,m}$ evaluates to $1$, we run another Grover's search with bounded-error inputs \cite{HMDw03} over the blocks.
    That in total requires $O\of*{\sqrt{m} \cdot \sqrt{m}} = O(m)$ queries.
    If the input $x$ is a $0$-input, this procedure will find the block $i$ where $\Exact_{1,m}$ evaluates to $1$ with constant success probability.
    To certify this, we then can classically query all variables in the $i$-th block to check whether there is a single $1$.
    Thus, with $O(m)$ queries, we will output $0$ on $0$-inputs with constant probability, and never output $0$ on $1$-inputs.    

{ Finally, this part immediately implies $\qqc\of*{f_m} = O(m)$.}
\end{proof}
}

{ As immediate  corollaries, we get that the same function separates 
$\QC$ from $\qqc_0^\cert$, $\adeg $ from $ \osdeg$, and
$\adeg^\cert $ from $ \osdeg^\cert$.}

\section{Conclusion and further work}

Systematically studying certification for various measures has allowed us to identify new bounds and to uncover new relationships between these and previously studied bounds.

There are many avenues to continue exploring.  We provided one separation, between degree certification complexity and rational degree.
We would also like to see other separations, or collapses.

Many polynomial relations between known bounds are not known to be tight. Can using certification measures help provide tighter bounds or bigger separations?

Finally, many relations between measures behave  differently for total and partial functions. They can hold for total functions, and be constant for some partial functions. What can be said about these relations for the class of  certifying functions? We hope that studying this class of partial functions will help us to gain a better understanding of why  this gap between total and partial functions can occur.

\section*{AI Disclosure}
We used GPT‑5.6 Terra, GPT‑5.6 Sol and Claude Opus 5 to assist with developing the proofs of \Cref{thm:separation (deg cert)^2.5 < ndeg}, \Cref{lem:selective}, \Cref{lem:poly} and \Cref{thm:Q_0 algo} in \Cref{sec: ndeg degcert separation}, as well as \Cref{range-reduction: exact}, \Cref{range-reduction: approximate} and \Cref{range-reduction: onesided} in \Cref{appendix:boundedness}. The authors verified the correctness and originality of all content including references. 

\section*{Acknowledgments}
E.L. has received support under the program ``Investissement d'Avenir'' launched by the French Government and implemented by ANR, with the reference ``ANR‐22‐CMAS-0001, QuanTEdu-France''. C.K. and E.L. are supported by the French PEPR integrated project EPiQ 
(ANR-22-PETQ-0007) and partially supported by the ANR Grant FLITTLA (ANR-21-CE48-0023). 
J.V.~is supported by the 1.1.1.9 Research application No 1.1.1.9/LZP/2/25/206 of the Activity ``Post-doctoral Research'' ``Practical applications and limitations of quantum search''.

\newpage
\printbibliography

@inproceedings{AA15,
    author = {Aaronson, Scott and Ambainis, Andris},
    title = {Forrelation: A Problem that Optimally Separates Quantum from Classical Computing},
    year = {2015},
    isbn = {978-1-45033-536-2},
    publisher = {Association for Computing Machinery},
    address = {New York, NY, USA},
    doi = {10.1145/2746539.2746547},
    booktitle = {Proceedings of the Forty-Seventh Annual ACM Symposium on Theory of Computing},
    pages = {307–316},
    numpages = {10},
    location = {Portland, Oregon, USA},
    series = {STOC '15},

    eprint    = {1411.5729},
    archivePrefix = {arXiv},
    primaryClass  = {quant-ph}
}

@inproceedings{AAI+16,
    author = {Aaronson, Scott and Ambainis, Andris and Iraids, J{\=a}nis and Kokainis, Martins and Smotrovs, Juris},
    title = {Polynomials, quantum query complexity, and Grothendieck's inequality},
    year = {2016},
    isbn = {978-3-95977-008-8},
    volume =	{50},
    editor =	{Raz, Ran},
    publisher =	{Schloss Dagstuhl -- Leibniz-Zentrum f{\"u}r Informatik},
    address =	{Dagstuhl, Germany},
    booktitle = {Proceedings of the 31st Conference on Computational Complexity},
    articleno = {25},
    pages =	{25:1--25:19},
    numpages = {19},
    series = {CCC '16},
    ISSN =	{1868-8969},
    doi =		{10.4230/LIPIcs.CCC.2016.25}
}

@inproceedings{HMDw03,
    author = {H\o{}yer, Peter and Mosca, Michele and {de} Wolf, Ronald},
    title = {Quantum Search on Bounded-Error Inputs},
    year = {2003},
    isbn = {3540404937},
    booktitle = {Automata, Languages and Programming},
    publisher = {Springer-Verlag},
    address = {Berlin, Heidelberg},
    pages = {291–299},
    numpages = {9},
    venue = {Eindhoven, The Netherlands},
    doi = {10.1007/3-540-45061-0_25},
    series = {ICALP 2003}
}

@article{Hoy00,
  title = {Arbitrary phases in quantum amplitude amplification},
  author = {H{\o}yer, Peter},
  journal = {Phys. Rev. A},
  volume = {62},
  issue = {5},
  pages = {052304},
  numpages = {5},
  year = {2000},
  publisher = {American Physical Society},
  doi = {10.1103/PhysRevA.62.052304},
}

@article{ABBL+17,
  author    = {Ambainis, Andris and
               Balodis, Kaspars and
               Belovs, Aleksandrs and
               Lee, Troy and
               Santha, Miklos and
               Smotrovs, Juris},
  title     = {Separations in Query Complexity Based on Pointer Functions},
  journal   = {Journal of the {ACM}},
  volume    = {64},
  number    = {5},
  pages     = {32:1--32:24},
  year      = {2017},
  doi       = {10.1145/3106234}
}

@misc{Jeffery2026QuantumQuery,
  author       = {Stacey Jeffery},
  title        = {Week 5: Quantum Query Complexity},
  year         = {2026},
  month        = mar,
  note         = {Lecture notes},
  howpublished = {\url{https://homepages.cwi.nl/~jeffery/notes/week5.pdf}}
}

@article{ABP19,
  author  = {Arunachalam, Srinivasan and Bri{\"e}t, Jop and Palazuelos, Carlos},
  title   = {Quantum Query Algorithms Are Completely Bounded Forms},
  journal = {SIAM Journal on Computing},
  volume  = {48},
  number  = {3},
  pages   = {903--925},
  year    = {2019},
  doi     = {10.1137/18M117563X}
}

@article{BBCM+01,
    author = {Beals, Robert and
               Buhrman, Harry and
               Cleve, Richard and
               Mosca, Michele and
               {de} Wolf, Ronald},
    title = {Quantum lower bounds by polynomials},
    year = {2001},
    issue_date = {July 2001},
    publisher = {Association for Computing Machinery},
    address = {New York, NY, USA},
    volume = {48},
    number = {4},
    issn = {0004-5411},
    url = {https://doi.org/10.1145/502090.502097},
    doi = {10.1145/502090.502097},
    journal = {J. {ACM}},
    month = jul,
    pages = {778–797},
    numpages = {20}
}

@phdthesis{Cade20,
  author = {Chris Cade},
  title  = {Postselection and Rational Functions in Quantum Complexity Theory},
  school = {University of Bristol},
  year   = {2020}
}

@article{Rub95,
  author    = {Rubinstein, David},
  title     = {Sensitivity vs. Block Sensitivity of {Boolean} Functions},
  journal   = {Combinatorica},
  volume    = {15},
  number    = {2},
  pages     = {297--299},
  year      = {1995},
  doi       = {10.1007/BF01200762}
}

@article{Hua19,
  title={Induced subgraphs of hypercubes and a proof of the sensitivity conjecture},
  author={Huang, Hao},
  journal={Annals of Mathematics},
  volume={190},
  number={3},
  pages={949--955},
  year={2019},
  doi = {10.4007/annals.2019.190.3.6},
}

@article{GSS16,
  author    = {Gilmer, Justin and
               Saks, Michael E. and
               Srinivasan, Srikanth},
  title     = {Composition limits and separating examples for some {Boolean} function
               complexity measures},
  journal   = {Combinatorica},
  volume    = {36},
  number    = {3},
  pages     = {265--311},
  year      = {2016},
  doi       = {10.1007/s00493-014-3189-x}
}

@article{BdW02,
  author    = {Buhrman, Harry and
               {de} Wolf, Ronald},
  title     = {Complexity measures and decision tree complexity: a survey},
  journal   = {Theoretical Computer Science},
  volume    = {288},
  number    = {1},
  pages     = {21--43},
  year      = {2002},
  doi       = {10.1016/S0304-3975(01)00144-X}
}

@inproceedings{ABK+21,
  author    = {Aaronson, Scott and
               Ben{-}David, Shalev and
               Kothari, Robin and
               Rao, Shravas and
               Tal, Avishay},
  title     = {Degree vs. approximate degree and Quantum implications of {H}uang's
               sensitivity theorem},
  booktitle = {{STOC}},
  pages     = {1330--1342},
  year      = {2021},
  doi       = {10.1145/3406325.3451047}
}

@article{NS94,
  author    = {Nisan, Noam and
               Szegedy, Mario},
  title     = {On the Degree of {Boolean} Functions as Real Polynomials},
  journal   = {Computational Complexity},
  volume    = {4},
  pages     = {301--313},
  year      = {1994},
  doi       = {10.1007/BF01263419}
}

@article{Aar08,
  author    = {Aaronson, Scott},
  title     = {Quantum certificate complexity},
  journal   = {Journal of Computer and System Sciences},
  volume    = {74},
  number    = {3},
  pages     = {313--322},
  year      = {2008},
  doi       = {10.1016/j.jcss.2007.06.020},
}

@inproceedings{Tal13,
  author    = {Tal, Avishay},
  title     = {Properties and applications of {Boolean} function composition},
  booktitle = {{ITCS}},
  pages     = {441--454},
  year      = {2013},
  doi       = {10.1145/2422436.2422485}
}

@article{KT16,
  author    = {Kulkarni, Raghav and
               Tal, Avishay},
  title     = {On Fractional Block Sensitivity},
  journal   = {Chicago Journal of Theoritical Computer Science},
  volume    = {2016},
  year      = {2016},
  url       = {http://cjtcs.cs.uchicago.edu/articles/2016/8/contents.html},
  bibsource = {dblp computer science bibliography, https://dblp.org}
}

@InProceedings{ABK21,
  author		=	{Anshu, Anurag and Ben{-}David, Shalev and Kundu, Srijita},
  title		=	{On Query-To-Communication Lifting for Adversary Bounds},
  booktitle	=	{36th Computational Complexity Conference (CCC 2021)},
  pages		=	{30:1--30:39},
  series		=	{Leibniz International Proceedings in Informatics (LIPIcs)},
  ISBN		=	{978-3-95977-193-1},
  ISSN		=	{1868-8969},
  year		=	{2021},
  volume		=	{200},
  editor		=	{Kabanets, Valentine},
  publisher	=	{Schloss Dagstuhl -- Leibniz-Zentrum f{\"u}r Informatik},
  address		=	{Dagstuhl, Germany},
  URL			=	{https://drops.dagstuhl.de/entities/document/10.4230/LIPIcs.CCC.2021.30},
  URN			=	{urn:nbn:de:0030-drops-143042},
  doi			=	{10.4230/LIPIcs.CCC.2021.30}
}

@article{JKK+20,
  author       = {Jain, Rahul and
                  Klauck, Hartmut and
                  Kundu, Srijita and
                  Lee, Troy and
                  Santha, Miklos and
                  Sanyal, Swagato and
                  Vihrovs, Jevg{\=e}nijs},
  title        = {Quadratically Tight Relations for Randomized Query Complexity},
  journal      = {Theory of Computing Systems},
  volume       = {64},
  number       = {1},
  pages        = {101--119},
  year         = {2020},
  date         = {2020/01/01},
  url          = {https://doi.org/10.1007/s00224-019-09935-x},
  doi          = {10.1007/s00224-019-09935-x},
  issn         = {1433-0490}
}

@inproceedings{CGL+23,
  author =	{Chakraborty, Sourav and G\'{a}l, Anna and Laplante, Sophie and Mittal, Rajat and Sunny, Anupa},
  title =	{{Certificate Games}},
  booktitle =	{14th Innovations in Theoretical Computer Science Conference (ITCS 2023)},
  pages =	{32:1--32:24},
  series =	{Leibniz International Proceedings in Informatics (LIPIcs)},
  ISBN =	{978-3-95977-263-1},
  ISSN =	{1868-8969},
  year =	{2023},
  volume =	{251},
  editor =	{Tauman Kalai, Yael},
  publisher =	{Schloss Dagstuhl -- Leibniz-Zentrum f{\"u}r Informatik},
  address =	{Dagstuhl, Germany},
  URL =		{https://drops.dagstuhl.de/entities/document/10.4230/LIPIcs.ITCS.2023.32},
  URN =		{urn:nbn:de:0030-drops-175353},
  doi =		{10.4230/LIPIcs.ITCS.2023.32}
}

@inproceedings{Nis89,
  author = {Nisan, Noam},
  title = {CREW PRAMS and decision trees},
  year = {1989},
  isbn = {0897913078},
  publisher = {Association for Computing Machinery},
  address = {New York, NY, USA},
  url = {https://doi.org/10.1145/73007.73038},
  doi = {10.1145/73007.73038},
  booktitle = {Proceedings of the Twenty-First Annual ACM Symposium on Theory of Computing},
  pages = {327–335},
  numpages = {9},
  location = {Seattle, Washington, USA},
  series = {STOC '89}
}

@InProceedings{BK19,
  author =	{Ben{-}David, Shalev and Kothari, Robin},
  title =	{{Quantum Distinguishing Complexity, Zero-Error Algorithms, and Statistical Zero Knowledge}},
  booktitle =	{14th Conference on the Theory of Quantum Computation, Communication and Cryptography (TQC 2019)},
  pages =	{2:1--2:23},
  series =	{Leibniz International Proceedings in Informatics (LIPIcs)},
  ISBN =	{978-3-95977-112-2},
  ISSN =	{1868-8969},
  year =	{2019},
  volume =	{135},
  editor =	{van Dam, Wim and Man\v{c}inska, Laura},
  publisher =	{Schloss Dagstuhl -- Leibniz-Zentrum f{\"u}r Informatik},
  address =	{Dagstuhl, Germany},
  URL =		{https://drops.dagstuhl.de/entities/document/10.4230/LIPIcs.TQC.2019.2},
  URN =		{urn:nbn:de:0030-drops-103944},
  doi =		{10.4230/LIPIcs.TQC.2019.2}
}

@article{dW03,
    author = {{de} Wolf, Ronald},
    title = {Nondeterministic Quantum Query and Communication Complexities},
    journal = {SIAM Journal on Computing},
    volume = {32},
    number = {3},
    pages = {681-699},
    year = {2003},
    doi = {10.1137/S0097539702407345},
    URL = {https://doi.org/10.1137/S0097539702407345},
    eprint = {https://doi.org/10.1137/S0097539702407345}
}

@incollection{BPM+02,
    author = {Brassard, Gilles and Høyer, Peter and Mosca, Michele and Tapp, Alain},
    series = {Contemporary Mathematics},
    issn = {0271-4132},
    pages = {53--74},
    volume = {305},
    publisher = {American Mathematical Society},
    booktitle = {Quantum Computation and Information},
    isbn = {978-0-82182-140-4},
    year = {2002},
    title = {Quantum amplitude amplification and estimation},
    language = {eng},
    address = {Providence, Rhode Island}
}

@InProceedings{CMP24,
  author =	{Cornelissen, Arjan and Mande, Nikhil S. and Patro, Subhasree},
  title =	{Quantum Sabotage Complexity},
  booktitle =	{44th IARCS Annual Conference on Foundations of Software Technology and Theoretical Computer Science (FSTTCS 2024)},
  pages =	{19:1--19:20},
  series =	{Leibniz International Proceedings in Informatics (LIPIcs)},
  ISBN =	{978-3-95977-355-3},
  ISSN =	{1868-8969},
  year =	{2024},
  volume =	{323},
  editor =	{Barman, Siddharth and Lasota, S{\l}awomir},
  publisher =	{Schloss Dagstuhl -- Leibniz-Zentrum f{\"u}r Informatik},
  address =	{Dagstuhl, Germany},
  URL =		{https://drops.dagstuhl.de/entities/document/10.4230/LIPIcs.FSTTCS.2024.19},
  URN =		{urn:nbn:de:0030-drops-222082},
  doi =		{10.4230/LIPIcs.FSTTCS.2024.19}
}

@article{AKP+21,
author = {Ambainis, Andris and Kokainis, Martins and Pr{\=u}sis, Kri{\v s}j{\=a}nis and Vihrovs, Jevg{\=e}nijs and Zajakins, Aleksejs},
title = {All Classical Adversary Methods Are Equivalent for Total Functions},
year = {2021},
issue_date = {March 2021},
publisher = {Association for Computing Machinery},
address = {New York, NY, USA},
volume = {13},
number = {1},
issn = {1942-3454},
url = {https://doi.org/10.1145/3442357},
doi = {10.1145/3442357},
journal = {ACM Trans. Comput. Theory},
month = jan,
articleno = {7},
numpages = {20}
}

@inproceedings{JK10,
author = {Jain, Rahul and Klauck, Hartmut},
title = {The Partition Bound for Classical Communication Complexity and Query Complexity},
year = {2010},
isbn = {978-0-76954-060-3},
year={2010},
publisher = {IEEE Computer Society},
address = {USA},
url = {https://doi.org/10.1109/CCC.2010.31},
doi = {10.1109/CCC.2010.31},
booktitle = {Proceedings of the 2010 IEEE 25th Annual Conference on Computational Complexity},
pages = {247–258},
numpages = {12},
series = {CCC '10}
}

@article{Bernstein12,
author="Bernstein, Sergei Natanovich",
title="Démonstration du théoréme de Weierstrass fondée sur le calcul des probabilités",
journal="Communications of the Kharkov Mathematical Society",
year="1912",
volume="2",
number="13",
pages="1"
}

@book{lorentz1986bernstein,
  title={Bernstein Polynomials},
  author={Lorentz, George G.},
  isbn={978-0-82840-323-8},
  lccn={85072466},
  series={AMS Chelsea Publishing Series},
  year={1986},
  publisher={Chelsea Publishing Company}
}

@article{midrijanis2004exact,
  title={Exact quantum query complexity for total Boolean functions},
  author={Midrijanis, Gatis},
  journal={arXiv preprint quant-ph/0403168},
  year={2004}
}

@inproceedings{CMS01,
    author = {Clementi, Andrea E. F. and Monti, Angelo and Silvestri, Riccardo},
    title = {Selective families, superimposed codes, and broadcasting on unknown radio networks},
    year = {2001},
    isbn = {0898714907},
    publisher = {Society for Industrial and Applied Mathematics},
    booktitle = {Proceedings of the Twelfth Annual ACM-SIAM Symposium on Discrete Algorithms},
    pages = {709–718},
    doi = {10.5555/365411.365756}
}

@phdthesis{Hra21,
    title={Efficient communication algorithms in shared channels with adversary},
    author={Hradovich, Ilya},
    year={2021},
    school={Institute of Computer Science, Polish Academy of Sciences},
    url={https://ipipan.waw.pl/pliki/doktoraty/Hradovich/Hradovich_Doktorat.pdf},
    type = {PhD Thesis}
}

@INPROCEEDINGS{BCdWZ99,
    author={Buhrman, Harry and Cleve, Richard and {de} Wolf, Ronald and Zalka, Christof},
    booktitle={40th Annual Symposium on Foundations of Computer Science}, 
    title={Bounds for small-error and zero-error quantum algorithms}, 
    year={1999},
    pages={358-368},
    doi={10.1109/SFFCS.1999.814607},
    eprint={cs/9904019},
    archivePrefix={arXiv},
    primaryClass={cs.CC},
}

@inproceedings{KKW+26,
  author    = {Kothari, Robin and
               Kovacs-Deak, Matt and
               Wang, Daochen and
               Yang, Rain Zimin},
  title     = {Rational degree is polynomially related to degree},
  booktitle = {{FOCS}},
  year      = {2026},
  note      = {to appear}
}

@article{Kou93,
  title={Improvements on Khrapchenko's theorem},
  author={Koutsoupias, Elias},
  journal={Theoretical Computer Science},
  volume={116},
  number={2},
  pages={399--403},
  year={1993},
  publisher={Elsevier}
}

@article{MdW15,
    author = {Mahadev, Urmila and {de} Wolf, Ronald},
    title = {Rational approximations and quantum algorithms with postselection},
    year = {2015},
    issue_date = {March 2015},
    publisher = {Rinton Press, Incorporated},
    address = {Paramus, NJ},
    volume = {15},
    number = {3–4},
    issn = {1533-7146},
    journal   = {Quantum Information \& Computation},
    month = mar,
    pages = {295–307},
    numpages = {13},
    doi       = {10.26421/QIC15.3-4-5},
    url       = {https://doi.org/10.26421/QIC15.3-4-5}
}

\newpage

\appendix
\section{Nondeterminism vs certification}\label{nondeterminism}

We now give the $\NP$-like definition of certification complexity (\Cref{def:cert qc witness version}) -- which is only defined on complexity measures that correspond to the query complexity of an algorithm, i.e. complexity measures coming from a model of computation (\Cref{def:query model of computation}). On such query complexity measures, this definition is equivalent to the general one above (\Cref{equivalence of cert qc versions}).

\begin{defi}[query model of computation]\label{def:query model of computation}
    A \emph{query model of computation} $\mathsf{M}$ specifies:
    \begin{itemize}
        \item what are the possible \emph{algorithms} (their operations, possibly some constraints on their structure)
        \item what it means for an algorithm $A$ to \emph{compute} an output $b$ when given query access to an input
        $x \in \zone^n$, i.e. $A^x = b$
        \item a corresponding \emph{query complexity} $\mathsf{M}(A,x)$ of algorithm $A$ on input $x$.
    \end{itemize}
    Thus, for a Boolean (possibly partial) function
    $f : \D \subseteq \zone^n \to \zone$,
    it gives a \emph{query complexity measure} of the function:
    $$\mathsf{M}(f) := \max\limits_{\substack{A\predicate{ algorithm in model } \mathsf{M}:\\A \predicate{ computes } f}} \  \max\limits_{x \in \D} \mathsf{M}(A,x)$$
    where $A$ \emph{computes} $f$ iff for all
    $x \in \D$
 computes $f(x)$ on input $x$.
\end{defi}

The the $\NP$-like definition of certification complexity is one of the motivations of our systematic study of the general notion.

\begin{defi}[query certification complexity $\mathsf{M}^\cert$ (witness version)]\label{def:cert qc witness version}
    Consider a query model of computation $\mathsf{M}$.
    Let
    $f : \D \subseteq \zone^n \to \zone$
    be a Boolean (possibly partial) function. Let
    $\C = \of*{w_x}_{x \in \D}$
    be a family of \emph{witnesses}. (Intuitively $w_x$ will be chosen to be a proof of knowledge of the value of $f(x)$, i.e. the correct proof $w_x$ should be accepted, and any wrong proof $w_y$ such that $f(x) \neq f(y)$ should be rejected. Thus, $\C$ certifies $f$, and we are interested in the best achievable query complexity of such a certification procedure.)
    
    In the context of the query model of computation $\mathsf{M}$, for any input
    $x \in \D$,
    we say that algorithm $A$ \emph{certifies $f$ at $x$ w.r.t. $\C$} with query complexity at most $T$ if
    \begin{itemize}
        \item certification soundness: the computation of $A$ with query access to input $x$ and plain access (for free, and classical) to witness $w_x$ accepts, i.e.
        $A^x\of*{w_x} = 1$,
        and query complexity of $A\of*{w_x}$ on $x$ is at most $T$, i.e.
        $\mathsf{M}\of*{A\of*{w_x},x} \leq T$
        \item certification completeness: for any input
        $y \in \D$
        such that
        $f(x) \neq f(y)$,
        the computation of $A$ with query access to input $x$ and plain access (for free, and classical) to witness $w_y$ (a wrong witness) rejects, i.e.
        $A^x\of*{w_y} = 0$,
        and query complexity of $A\of*{w_y}$ on $x$ is at most $T$, i.e.
        $\mathsf{M}\of*{A\of*{w_y},x} \leq T$.
    \end{itemize}
    The \emph{query $f$-certification complexity of $A$ w.r.t. $\C$} is
    $$\mathsf{M}^\cert(A,\C,f) := \max\limits_{x \in \D}\ \min\limits_{\substack{T \in \NN:\\A \predicate{ certifies } f \predicate{ at } x \predicate{ w.r.t. } \C\\\predicate{with query complexity at most }T}}\ T.$$
    Thus, the \emph{query certification complexity of $f$} in model $\mathsf{M}$ is
    $$\mathsf{M}^\cert(f) := \min\limits_{\C}\ \min\limits_{\substack{A \predicate{ algorithm in model } \mathsf{M}:\\\for{x \in \D} A \predicate{ certifies } f \predicate{ at } x \predicate{ w.r.t. } \C}}\ \mathsf{M}^\cert(A,\C,f).$$
\end{defi}

Saying that algorithm $A$ certifies $f$ at $x$ w.r.t. $\C$ only specifies how $A$ behaves with query access to an input in $\set{z} \cup \comprehension{y\in\D}{f(x) \neq f(y)}$ and nothing outside of this domain. Restrictions on the behavior of the algorithm outside of this domain fall within the specification of the query model of computation. More generally (\Cref{def:cert qc certification version}), one should be aware that the complexity measure $\mathsf{M}$ should be defined on partial functions for $\mathsf{M}^\cert$ to be defined (event on total functions).

\begin{lemma}[equivalence of the definitions of query certification complexity]\label{equivalence of cert qc versions}
    \Cref{def:cert qc certification version} and \Cref{def:cert qc witness version} define the same query certification complexity, i.e. for any query model of computation $\mathsf{M}$ and any Boolean (possibly partial) function
    $f : \D \subseteq \zone^n \to \zone$
    they define the same complexity measure $\mathsf{M}^\cert(f)$.
\end{lemma}
\begin{proof}
    We label by $\text{w}$ objects and predicates defined by the witness version of query certification complexity (\Cref{def:cert qc witness version}) and by $\text{c}$ those defined by the usual certification version (\Cref{def:cert qc certification version}). Consider a query model of computation $\mathsf{M}$.

    Fist we show
    $\mathsf{M}^\cert_\text{c}(f) \leq \mathsf{M}^\cert_\text{w}(f)$.
    Let the family of witnesses
    $\C = \of*{w_z}_{z\in\D}$ and algorithm $A$ be such that for all
    $x \in \D$
    algorithm $A$ certifies $f$ at $x$ w.r.t. $\C$, and
    $\mathsf{M}^\cert_\text{w}(f) = \mathsf{M}^\cert_\text{w}(A,\C,f)$
    (i.e. they are optimal).
    So for any inputs
    $x,z\in\D$,
    if
    $z=x$
    then
    $A^x\of*{w_z}=1$
    with query complexity
    $\mathsf{M}(A\of*{w_z),x} \leq \mathsf{M}^\cert_\text{w}(A,\C,f) = \mathsf{M}^\cert_\text{w}(f)$,
    and if
    $f(z) \neq f(x)$
    then
    $A^x\of*{w_z}=0$
    with query complexity
    $\mathsf{M}(A\of*{w_z),x} \leq \mathsf{M}^\cert_\text{w}(A,\C,f) = \mathsf{M}^\cert_\text{w}(f)$.
    Let algorithm $B$ be such that for any inputs
    $x,z\in\zone^n$,
    $B^x(z)$ first computes $w_z$ (for free) and continues like $A^x\of*{w_z}$. Then, for all inputs
    $z \in \D$
    and
    $x \in \set*{z} \cup \comprehension{y \in \D}{f(y) \neq f(z)}$,
    $B^x(z)=A^x(z)=f^\cert_z(x)$
    with query complexity
    $\mathsf{M}(B(z),x) = \mathsf{M}\of*{A\of*{w_z},x} \leq \mathsf{M}^\cert_\text{w}(f)$.
    So
    $$\mathsf{M}^\cert_\text{c}(f)
    = \max\limits_{z\in\D} \mathsf{M}(f^\cert_z)
    \leq \max\limits_{z\in\D} \max\limits_{x\in\set*{z} \cup \comprehension{y \in \D}{f(y) \neq f(z)}} \mathsf{M}(B(z),x)
    \leq \max\limits_{z\in\D} \mathsf{M}^\cert_\text{w}(f)
    = \mathsf{M}^\cert_\text{w}(f).$$

    Now we show the other way
    $\mathsf{M}^\cert_\text{w}(f) \leq \mathsf{M}^\cert_\text{c}(f)$.
    For all
    $z\in\D$,
    let
    $A(z)$ be an algorithm computing $f^\cert_z$ with optimal query complexity $\mathsf{M}\of*{f^\cert_z}$. So for all inputs
    $z\in\D$
    and
    $x \in \set*{z} \cup \comprehension{y \in \D}{f(y) \neq f(z)}$,
    $A^x(z)=f^\cert_z(x)$
    with query complexity
    $\mathsf{M}(A(z),x) \leq \mathsf{M}\of*{f^\cert_z}$.
    Since $\D$ is finite, then $\of*{A(z)}_{z\in\D}$ is a uniform family of algorithms, so $A$ itself is an algorithm. So $A$ $\text{w}$-certifies $f$ w.r.t.
    $\C:=(z)_{z\in\D}$ with query complexity
    $\mathsf{M}^\cert_w(A,\C,f) \leq \max\limits_{z\in\D} \mathsf{M}\of*{f^\cert_z} = \mathsf{M}^\cert_\text{c}(f)$. Thus
    $$\mathsf{M}^\cert_\text{w}(f)
    \leq \mathsf{M}^\cert_w(A,\C,f)
    \leq \mathsf{M}^\cert_\text{c}(f).\qedhere$$
\end{proof}
%==================
\section{Proof of \Cref{C = (C_0 C_1) cert}}\label{appendix:proof of lemmas}

\begin{lemma}\label{C = (C_0 C_1) cert}
    For any boolean (possibly partial) function
    $f:\D \subseteq \zone^n \to \zone$
    and any
    $k > 0$,
    $$\of*{\cc_0\cdot\cc_1}^\cert(f) = \cc(f)$$
\end{lemma}
\begin{proof}
    Suppose $f$ is non-constant, as whenever $f$ is constant
    $\of*{\cc_0\cdot\cc_1}^\cert(f) = 0 = \cc(f)$
    and the equality holds.
    Since for a $0$-input $y$ of $f^\cert_z$, a single bit $i$ where $y_i$ differs from $z_i$ is enough to certify $y$ (and at least one bit is necessary to certify $y$ as $f$ is non-constant), then
    $\cc_0\of*{f^\cert_z} = 1$
    holds for any input
    $z \in \D$.
    Hence
    \begin{align*}
        \of*{\cc_0\cdot\cc_1}^\cert(f)
        &= \max\limits_{z\in\D} {\cc_0\of*{f^\cert_z}} \cdot \cc_1\of*{f^\cert_z}\\
        &= \max\limits_{z\in\D} 1 \cdot \max\of*{1 , \cc_1\of*{f^\cert_z}}\\
        &= \max\limits_{z\in\D} \max\of*{\cc_0\of*{f^\cert_z} , \cc_1\of*{f^\cert_z}}\\
        &= \cc^\cert(f)\\
        &= \cc(f)
    \end{align*}
    where the second equality follows from
    $\cc_0\of*{f^\cert_z} = 1$
    and
    $\cc_1\of*{f^\cert_z} \geq 1$
    as $f$ is non constant, the third equality follows from
    $\cc_0\of*{f^\cert_z} = 1$,
    and the last equality follows from \Cref{local cert = local}.
\end{proof}

\section{Independence of bounded-polynomial degree measures from the boundedness parameter}\label{appendix:boundedness}

{\em The proofs in Appendix~\ref{appendix:boundedness} were obtained using significant assistance from GPT‑5.6 Terra and Claude Opus 5.}

The following three lemmas allow us to show that up to constant factors in the degree, the bounding parameter for bounded polynomials can be taken to be 1.
Rescaling the polynomials does not increase the error.

\begin{lemma}[range reduction]\label{range-reduction: exact}
    Let
    $f: \D \subseteq \zone^n \to \zone$
    be a Boolean (possibly partial) function and let $P$ be a degree $d$ exact polynomial for $f$, with values bounded to the range $\bracket*{-C,C}$ for all points on the Boolean hypercube, where $C > 1$, i.e.
    \begin{align*}
        P(x)=f(x)  & \qquad \forall x \in \D, \\
        P(x) \in \bracket*{-C,C} & \qquad \forall x \in \zone^n ,\\
        \deg(P)=d.
    \end{align*}
    Then there exists a polynomial $Q$ that exactly computes $f$ with range $\bracket*{0,1}$ on the hypercube and degree $m_{C} \cdot d$, where $m_{C}$ is a constant depending only on $C$, i.e.
    \begin{align*}
        Q(x)=f(x) & \qquad \forall x \in \D, \\
        Q(x) \in \bracket*{0,1} & \qquad \forall x \in \zone^n ,\\
        \deg(Q) \leq m_{C} \cdot d.
    \end{align*}
\end{lemma}
\begin{proof}
    We will construct a univariate polynomial $A$ with
    $$
    A\of*{\bracket*{-C,C}} \subseteq \bracket*{0,1}, \qquad A(0)=0, \qquad A(1)=1.
    $$

    Choose
    $m \equiv 2 \pmod 4$
    such that
    $\pi C \leq m \leq \pi C + 4$. Then
    $$
    A(t) = \frac{1 + T_m\of*{\sin\of*{\frac{\pi}{m}} t}}{2}
    $$
    where $T_m$ is the degree-$m$ Chebyshev polynomial of the first kind, i.e. defined by the property that
    $T_m(\cos \theta)=\cos(m \theta)$.
    Denote
    $u(t) := \sin\of*{\frac{\pi}{m}} t$.

    Since
    $\sin\of*{\frac{\pi}{m}} \leq \frac{\pi}{m} \leq \frac{1}{C}$,
    we have for every
    $t \in \bracket*{-C,C}$
    $$
    \card*{t} \sin\of*{\frac{\pi}{m}} \leq C \sin\of*{\frac{\pi}{m}} \leq 1 ,
    $$
    so
    $u(t) \in \bracket*{-1,1}$.
    As
    $\card*{T_m} \leq 1$
    on
    $\bracket*{-1,1}$,
    it follows that
    $A(t) \in \bracket*{0,1}$
    for all
    $t \in \bracket*{-C,C}$.

    Since
    $m \equiv 2 \pmod 4$,
    then $\frac{m}{2}$ is an odd integer and
    $T_m(u(0))=T_m(0) = \cos\of*{\frac{m\pi}{2}} = -1 $,
    therefore
    $$
    A(0) = \frac{1 + (-1)}{2} = 0 .
    $$
    On the other hand,
    $u(1) = \sin\of*{\frac{\pi}{m}} = \cos\of*{\frac{\pi}{2} - \frac{\pi}{m}}
    = \cos\of*{\frac{\of*{\frac{m}{2} - 1}\pi}{m}}$,
    and
    $\frac{m}{2} - 1$
    is even, so
    $T_m(u(1)) = \cos\of*{\of*{\frac{m}{2} - 1}\pi} = 1$.
    Therefore 
    $$
    A(1) = \frac{1 + 1}{2} = 1 .
    $$

    Finally, define polynomial $Q$ as
    $Q(x):=A(P(x))$.
    This satisfies all the conditions, and has degree
    $\deg(Q) \leq m \cdot d$,
    where $m$ depends only on $C$.
\end{proof}

\begin{lemma}[approximate range reduction]\label{range-reduction: approximate}
    Let $f: \D \subseteq \zone^n \to \zone$ be a Boolean (possibly partial) function and let $P$ be an $\varepsilon$-approximating polynomial for $f$ with $0 < \varepsilon < \frac{1}{2}$ and degree $d$, with values bounded to the range $\bracket*{-C,C}$ for all points on the Boolean hypercube:
    \begin{align*}
        \card*{P(x)-f(x)} \leq \varepsilon & \qquad \for{x \in \D}\\
        P(x) \in \bracket*{-C,C} & \qquad \for{x \in \zone^n}\\
        \deg(P)=d.
    \end{align*}
    Then there exists a polynomial $Q$ that $\varepsilon$-approximates $f$ with range $\bracket*{0,1}$ on the hypercube and degree at most $m_{C,\varepsilon} \cdot d$, where $m_{C,\varepsilon}$ is a constant depending only on $C$ and $\varepsilon$:
    \begin{align*}
        \card*{Q(x)-f(x)} \leq \varepsilon & \qquad \for{x \in \D}\\
        Q(x) \in \bracket*{0,1} & \qquad \for{x \in \zone^n}\\
        \deg(Q) \leq m_{C,\varepsilon} \cdot d.
    \end{align*}
\end{lemma}
\begin{proof}
    Define the continuous function $g:\bracket*{-C,C} \to \bracket*{0,1}$:
    $$
        g(t)=
        \begin{cases}
            0 &\predicate{if } t\leq\varepsilon, \\
            \dfrac{t-\varepsilon}{1-2\varepsilon} 
            &\predicate{if } \varepsilon<t<1-\varepsilon, \\
            1 &\predicate{if } t\geq 1-\varepsilon.
        \end{cases}
    $$
    Then we can rescale $\bracket*{-C,C}$ to $\bracket*{0,1}$ and examine the $m$-th Bernstein approximation polynomial of this function:
    $$
    A_m(t):=
    \sum_{k=0}^m
    g\of*{-C+\frac{2Ck}{m}}
    \binom{m}{k}
    \of*{\frac{t+C}{2C}}^k
    \of*{1-\frac{t+C}{2C}}^{m-k}.
    $$
    Then, by the Bernstein approximation theorem \cite{Bernstein12}, there exists some $m_{C,\varepsilon}$ such that
    $\card*{ A_{m_{C,\varepsilon}}(t) - g(t) } \leq \varepsilon$ on $t \in \bracket*{-C,C}$.

    Furthermore, by the properties of Bernstein polynomials \cite{lorentz1986bernstein}, for
    $t\in\bracket*{-C,C}$,
    the quantities multiplying the values of $g$ are nonnegative and sum to $1$. Since
    $g(t) \in \bracket*{0,1}$,
    it follows that
    $A_{m_{C,\varepsilon}}(t)\in\bracket*{0,1}$
    for all
    $t \in\bracket*{-C,C}$.

    Now let $Q(x)=A_{m_{C,\varepsilon}}(P(x))$. Since $P(x)\in\bracket*{-C,C}$ for every $x\in\zone^n$, we have $Q(x)\in\bracket*{0,1}$
    for all $x\in\zone^n$.
 
    On the other hand, if
    $x\in \D$,
    we have that either
    $P(x)\in\bracket*{-\varepsilon,\varepsilon}$
    or 
    $P(x)\in\bracket*{1-\varepsilon,1+\varepsilon}$ 
    and hence
    $g(P(x))=0$
    or
    $g(P(x))=1$
    exactly. Therefore
    $$
    \card*{Q(x)-f(x)}
    =\card*{A_{m_{C,\varepsilon}}(P(x))-g(P(x))}
    \leq \varepsilon.
    $$
    Finally,
    $\deg(Q)\leq \deg(A) \cdot \deg(P)\leq m_{C,\varepsilon} \cdot d$.
\end{proof}

\begin{lemma}[one-sided range reduction]\label{range-reduction: onesided}
    Let $f: \D \subseteq \zone^n \to \zone$ be a Boolean (possibly partial) function and let $P$ be a one-sided-error $\varepsilon$-approximating polynomial for $f$ with $0 < \varepsilon < \frac{1}{2}$ and degree $d$, with values bounded to the range $\bracket*{-C,C}$ for all points on the Boolean hypercube.

    Then there exists a polynomial $Q$ that is a one-sided $\varepsilon$-approximation of  $f$ with range $\bracket*{0,1}$ on the hypercube and degree at most $m_{C,\varepsilon}\cdot d$, where $m_{C,\varepsilon}$ is a constant depending only on $C$ and $\varepsilon$.
\end{lemma}
\begin{proof}
    W.l.o.g. assume that $P(x)=0$ on $f^{-1}(0)$ and $P(x) \in \bracket*{1-\varepsilon,1+\varepsilon}$ on $f^{-1}(1)$.

    The proof extends the construction of \Cref{range-reduction: approximate}, define the function $g(t)$ and polynomial $A_m(t)$ as in that proof. However, we will now pick $m_{C,\varepsilon'}$ so that $\card*{ A_{m_{C,\varepsilon'}}(t) - g(t) } \leq \varepsilon'$, for $\varepsilon' < \varepsilon$ to be chosen later.

    Let $\delta:=A_{m_{C,\varepsilon'}}(0)$, by construction $\delta \leq \varepsilon'$.
    Then define the polynomial 
    $$
    B(t)=\of*{\frac {A_{m_{C,\varepsilon'}}(t)-\delta }{1-\delta}} ^2.
    $$
    For
    $t \in \bracket*{-C,C}$
    we have
    $\frac{A_{m_{C,\varepsilon'}}(t)-\delta }{1-\delta} \in \bracket*{-\frac{\delta}{1-\delta},1}$,
    therefore
    $B(t) \in \bracket*{0,1}$.

    For $t=0$ we have $B(t)=0$. 

    Finally, for $t\in \bracket*{1-\varepsilon,1+\varepsilon}$ we have that $A_{m_{C,\varepsilon'}}(t) \in \bracket*{1-\varepsilon', 1}$, therefore 
    $$
    B(t) \geq \of*{\frac {1-\varepsilon'-\delta}{1-\delta}}^2 \geq 
    (1-2\varepsilon')^2.
    $$
    Then we can pick $\varepsilon'$ so that this value is at least $1-\varepsilon$.

    Finally, take $Q(x)=B(P(x))$, obtaining a polynomial with degree $2 d \cdot m_{C,\varepsilon'}$.
\end{proof}

\end{document}